\documentclass[11pt]{article}
\usepackage[margin=1in]{geometry}
\usepackage[utf8]{inputenc}
\usepackage[T1]{fontenc}
\usepackage{microtype}
\usepackage{amsmath, amssymb, amsthm, dsfont}
\usepackage{enumitem}
\usepackage{graphicx, subfig}
\usepackage[dvipsnames]{xcolor}
\usepackage{hyperref, comment,url}
\usepackage[ruled,vlined,linesnumbered]{algorithm2e}

\usepackage[utf8]{inputenc}

\usepackage{hyperref}
\usepackage{color}
\usepackage{graphicx}
\graphicspath{ {./Images/} }
\usepackage{wrapfig}
\usepackage{amsmath, amsthm, amssymb}
\usepackage{xfrac}

\usepackage{theoremref}

\usepackage{xcolor}
\hypersetup{
    colorlinks,
    linkcolor={red!50!black},
    citecolor={blue!50!black},
    urlcolor={blue!80!black}
}

\newtheorem{theorem}{Theorem}[section]
   \newtheorem{thm}{Theorem}[section]
   
   \newtheorem{lemma}[thm]{Lemma}
   \newtheorem{corollary}[thm]{Corollary}
   
   \newtheorem{proposition}[thm]{Proposition}
   \newtheorem{definition}{Definition}
   \newtheorem{remark}{Remark}

\newcommand{\vs}{\vspace{3mm}}

\newcommand{\at}[2][]{#1|_{#2}}
\newcommand\Real{\textup{Re}}
\newcommand\poly{\textup{poly}}
\newcommand\ol[1]{\overline{#1}}
\newcommand\wt[1]{\widetilde{#1}}
\newcommand\N{\mathbb{N}}
\newcommand\bfx{\textbf{x}}
\newcommand\bfy{\textbf{y}}
\newcommand\bfz{\textbf{z}}
\newcommand\bfa{\textbf{a}}

\newcommand\bfj{\textbf{j}}
\newcommand\bfd{\textbf{d}}
\newcommand\bfw{\textbf{w}}

\newcommand\mcQ{\mathcal{Q}}
\newcommand\mcP{\mathcal{P}}
\newcommand\bad{\Xi_{\text{bad}}}
\newcommand\good{\Xi_{\text{good}}}
\newcommand\E{\mathbb{E}}

\newcommand\C{\mathbb{C}}
\newcommand\D{\mathbb{D}}

\newcommand\defeq{\stackrel{def}{=}}

\newcommand\abs[1]{{\left\lvert{#1}\right\rvert}}

\newcommand\supp{\textup{Supp}}

\usepackage{soul}
\usepackage{xcolor}

\usepackage{amsmath}
\makeatletter
\def\resetMathstrut@{%
  \setbox\z@\hbox{%
    \mathchardef\@tempa\mathcode`\[\relax
    \def\@tempb##1"##2##3{\the\textfont"##3\char"}%
    \expandafter\@tempb\meaning\@tempa \relax
  }%
  \ht\Mathstrutbox@\ht\z@ \dp\Mathstrutbox@\dp\z@}
\makeatother

\title{Average-Case to (shifted) Worst-Case Reduction for the Trace Reconstruction Problem}
\usepackage{authblk}

\author[1, 2]{Ittai Rubinstein}

\affil[1]{Blavatnik School of Computer Science, Tel-Aviv University, Tel-Aviv 69978, Israel}
\affil[2]{QEDMA Quantum Computing, Tel-Aviv, Israel}
\date{\today}

\begin{document}

\maketitle

\begin{abstract}
    The {\em insertion-deletion channel} takes as input a binary string $\bfx \in\{0, 1\}^n$, and outputs a string $\wt{\bfx}$ where some of the bits have been deleted and others inserted independently at random.
    In the {\em trace reconstruction problem}, one is given many outputs (called {\em traces}) of the insertion-deletion channel on the same input message $\bfx$, and asked to recover the input message.
    
    Nazarov and Peres, and De et al \cite{nazarov2017trace, de2017optimal} showed that any string $\bfx$ can be reconstructed from $\exp(O(n^{1/3}))$ traces.
    Holden et al \cite{holden2018subpolynomial} adapt the techniques used to prove this upper bound, to an algorithm for the average-case trace reconstruction with a sample complexity of $\exp(O(\log^{1/3} n))$.
    However, it is not clear how to apply their techniques more generally and in particular for the recent worst-case upper bound of $\exp(\widetilde{O}(n^{1/5}))$ shown by Chase \cite{chase2021separating} for the deletion-channel.
    
    We prove a general reduction from the average-case to smaller instances of a problem similar to worst-case.
    Using this reduction and a generalization of Chase's bound, we construct an improved average-case algorithm with a sample complexity of $\exp(\widetilde{O}(\log^{1/5} n))$.
    Additionally, we show that Chase's upper-bound holds for the insertion-deletion channel as well.

\end{abstract}

\setcounter{page}{0}

\thispagestyle{empty}

\newpage

\section{Introduction}

The {\em insertion-deletion channel} with parameters $0 \leq q, q^\prime < 1$ takes as input a string $\bfx \in \{0, 1\} ^ n$. 
For each $j$, $G_j$ random uniform and independent bits are inserted before the $j$th bit of $\bfx$, where the random variables $G_j \geq 0$ are i.i.d. geometrically distributed with parameter $q^\prime$. 
Then, each bit of the message is deleted independently with probability $q$.
The output string $\widetilde\bfx$ is called a {\em trace}.

The {\em trace reconstruction problem} asks the following question: how many traces are necessary to reconstruct an unknown string $\bfx$?

The main motivation for studying this problem comes from computational biology, where one often tries to align several DNA sequences to a common ancestor. 
It has been extensively researched since the early 2000's \cite{batu2004reconstructing}. 
Over the past few years, the trace reconstruction problem has received an increased focus, spawning many new versions, such as the coded trace reconstruction \cite{cheraghchi2020coded}, the approximate trace reconstruction \cite{chase2021approximate, chen2022near} and the population recovery and trace reconstruction problems \cite{ban2019efficient}.

In this paper we will focus on the two main versions introduced by Batu et al \cite{batu2004reconstructing}, called the {\em worst-case} and the {\em average-case}\footnote{Sometimes also called the ``random case".}.
In the worst-case, the message $\bfx$ is adversarially chosen, so the method used to reconstruct it must work for all strings $\bfx \in \{0, 1\} ^ n$.
In the average-case, $\bfx$ is a random string of bits and the reconstruction only needs to succeed with high probability over the choice of $\bfx$.

There appears to be an exponential gap between these cases.
Indeed, McGregor et al \cite{mcgregor2014trace} showed that if $h(n)$ traces are necessary for the worst-case trace reconstruction, then at least $h(\log{n})$ are needed for the average-case (and under some conditions $h(\log{n}) \log{n}$).
The best known lower bounds on the average-case have followed a similar pattern with Chase proving a lower-bound of $n^{3/2} / poly(\log{n})$ and $\log ^ {5/2} n / poly(\log{\log{n}})$ samples for the worst-case and average-case respectively \cite{chase2021new}, improving upon the previous bounds of $n^{5/4} / poly(\log{n})$ and $\log ^ {9/4} n / poly(\log\log{n})$ for the worst-case and average-case respectively by Holden and Lyons \cite{holden2020lower}.

The upper bounds have also followed a similar suit.
Holenstein et al \cite{holenstein2008trace} established an upper bound of $\exp(\widetilde{O}(n^{1/2}))$ on the sample complexity of the worst-case trace reconstruction problem. 
This was improved by Nazarov and Peres \cite{nazarov2017trace}, and De, O'Donnell and Servedio \cite{de2017optimal} who simultaneously proved that $\exp(O(n^{1/3}))$ traces are sufficient, and later by Chase \cite{chase2021separating} who improved the bound to $\exp(\widetilde{O}(n^{1/5}))$ for deletion channels (i.e. with $q^\prime = 0$).

Peres and Zhai \cite{peres2017average} adapted the $\exp(O(n^{1/3}))$ bound to the average-case, constructing an efficient algorithm for the average trace reconstruction with $\exp(O(\log ^ {{1}/{2}} n))$ samples and with some limitations on the deletion probability ($q \leq 1/2$ and $q^\prime = 0$).
This was further improved by Holden et al \cite{holden2020subpolynomial} who reduced the sample complexity to $\exp(O(\log ^ {{1}/{3}} n))$ and generalized the algorithm to work for all insertion-deletion channels.

\subsection{An Overview of Previous Results}
Our results are mainly based on an adaptation and a combination of the techniques used in Holden et al and Chase's papers \cite{holden2020subpolynomial, chase2021separating}.
Here, we will give a brief overview of their methods and why it is not trivial to combine them.

\subsubsection{An Overview of \texorpdfstring{\cite{holden2020subpolynomial}}{HPPZ20}}

Holden et al present an algorithm for the average-case trace reconstruction.
This algorithm reconstructs the string $\bfx$ one bit at a time.

In the $k$th iteration, an alignment procedure is run with the goal of matching an index $k^\prime \approx k - C \log n$ slightly less than $k$ in the original message to an index $\tau_{k^\prime}$ in each of the traces.
This alignment is noisy, resulting in a small random shift and occasionally a completely missed alignment.

After this alignment, the bits following each aligned index $\tau_{k^\prime}$ are viewed as a trace of the bits immediately after $k^\prime$ in the original string $\bfx$.

Reconstructing the $k$th bit from these new traces presents several new difficulties.
First, one must deal with spurious matches (cases where the alignment failed completely).
Then one must deal with the fact that even in the ideal scenario, the alignments are not precise.
Finally, instead of reconstructing a string of some given length $\log n$, we reconstruct the first $\log n$ bits of a far longer string.

Holden et al then show that the complex analysis techniques used for the worst-case bounds by \cite{nazarov2017trace} and \cite{de2017optimal} can be adapted to this new problem and to insertion-deletion channels.

Roughly speaking, these techniques work by converting a function of the traces to a polynomial that depends on the original message.
This polynomial is then shown to have a sufficiently strong dependence on the $k$th bit of $\bfx$, when evaluated at some point where it can be approximated from a sufficiently small number of traces.

\subsubsection{An Overview of \texorpdfstring{\cite{chase2021separating}}{Cha21b}}
The $\exp(O(n^{1/3}))$ upper-bounds on the sample complexity of the worst-case trace reconstruction and most similar bounds are shown using a mean-based algorithm -- an algorithm that considers the distribution of the $i$th bit of the traces separately for each $i$ \cite{nazarov2017trace, de2017optimal}.
However, these same papers also show a matching lower bound for mean-based algorithms.

In order to overcome this, Chase showed that separators which are based on highly non-linear functions of the traces (and are thus not mean-based algorithms), can be used for the worst-case trace reconstruction problem, reducing the sample complexity to $\exp(\wt{O}(n^{1/5}))$ \cite{chase2021separating}.
Analysing these separators requires an extension of the complex analysis used by Nazarov and Peres and by De et al to the multivariate case.

\subsubsection{Combining these Results}

The first difficulty in combining these results is the need to apply Holden et al's bit recovery procedure in a different context.
This bit recovery procedure has many parameters, which were defined in \cite{holden2020subpolynomial} only for the specific case of their algorithm.
While it is not exceptionally difficult to adapt it to other scenarios, it does require a long and technical proof.
In order to make these techniques more accessible to future researchers, we convert them into a general reduction.

The next difficulty is that Holden et al's conversion of the bit recovery procedure to a complex analysis problem depends on the fact that they use a mean-based separator.
When analysing such separators, many terms relating to the insertions of the channel cancel out.
However, the main advantage of Chase's upper bound cannot be obtained using such separators.

Finally, in the complex analysis itself, one obtains a geometric sum related to the traces which can be used to estimate the values of some polynomial related to the original message $\bfx$.
This polynomial and a point in which to evaluate it are carefully chosen so that they will have a strong dependency on $\bfx$ and that the geometric summation will not be too large.

However, in the worst-case analysis, this technique is used when the summation is truncated by the length of the trace, allowing one to compute it at points where it might not converge.
Because the bit recovery attempts to reconstruct the prefix of a very long string, it no longer suffices to show that this geometric series grows slowly.
In fact, we need to show that it decays rapidly so that it can be truncated.
This requires many changes to the method by which the evaluation point is selected and its analysis.

\subsection{Our Contribution}
Our main contribution is an improvement of Holden et al's algorithm \cite{holden2020subpolynomial}, with a sample complexity of $\exp(\widetilde{O}(\log ^ {{1}/{5}} n))$:

\begin{theorem} [Main Result]
\thlabel{thm:main_result}
    For any constant parameters $q, q^\prime \in [0, 1)$, there exists $M > 0$, such that for any $n \in \N$, if $\bfx \in \{0, 1\}^n$ is a bit-string where the bits are chosen uniformly and independently at random, then we can reconstruct $\bfx$ with probability $1 - o_n(1)$ using $\lceil \exp(M \log^{1/5} n \log^7 \log n) \rceil$ traces from the insertion-deletion channel with parameters $q, q^\prime$.
    Moreover, when $q < 1/2$ this can be done in $n^{1+o(1)}$ time and when $q \geq 1/2$, this can be done in polynomial time.
\end{theorem}

To show this, we introduce a new version of the trace reconstruction problem similar to the worst-case, which we call the {\em shifted trace reconstruction} (see Definition \ref{def:shifted_trp}).
The shifted trace reconstruction problem differs from the worst-case trace reconstruction in three key ways:

First, some ($o(1)$ fraction) of the samples may be ``false samples" - adversarially chosen strings mixed into our pool of samples.
These false samples arise in our reduction, because we will perform our alignment with very short substrings, so some of our alignments will come from spurious matches.
However, we do not expect the addition of a sufficiently small fraction of false samples to have a significant effect on the difficulty of trace reconstruction.
This is because the information-theoretically optimal separation between the traces of two possible strings $\bfx$ and $\bfy$ would be done using a likelihood estimation and this can be easily amended to deal with a small fraction of incorrect entries.

The second difference is that the traces are shifted (owing to the inaccuracy of our alignment procedure), in the sense that there is some $\N$ valued random variable $S$ (bounded on some interval of length $\eta$), and before applying the channel to produce any single trace we will erase its first $S$ bits.
While this could potentially make the reconstruction harder, in practice the complex analysis based reconstruction techniques \cite{nazarov2017trace, de2017optimal, chase2021separating} which are most commonly used for the worst-case, can be fairly easily adapted.

The final and largest difference is that instead of reconstructing a string $\bfx$ of finite length $n$, in the shifted reconstruction, we are given traces of a very long string (which can be thought of as exponentially or infinitely long), and are asked to reconstruct the first $n$ bits.
Dealing with the longer strings requires many changes to the complex analysis techniques used for the worst-case reconstruction.
Peres and Zhai and Holden et al make these adaptions to the mean-based analyses in their reconstruction algorithms \cite{peres2017average, holden2020subpolynomial} but similarly adapting Chase's construction is not trivial.

We adapt Holden et al's methods \cite{holden2020subpolynomial}, which were originally used for a specific average-case algorithm, to create a general reduction from the average-case trace reconstruction to (a smaller version of) the shifted trace reconstruction.
This reduction (\thref{thm:reduction}) can be used in order to convert additional advances on the worst-case trace reconstruction problem into efficient algorithms for the average-case, and indeed we use it to improve Holden et al's algorithm. 

\begin{theorem} [Average to Shifted Reduction]
\thlabel{thm:reduction}
    For any constant $q, q^\prime \in [0, 1)$, and any positive $C_1 > 0$, there exists some positive constant $C_2 > 0$, 
    such that:
    
    For any monotone function $\log(n) \leq h(n) \leq \sqrt{n}$ and any algorithm $A$ that solves the shifted trace reconstruction problem with 
    sample complexity $\sigma(n) \leq \exp(h(n))$, 
    time complexity $t(n)$, 
    false sample rate $\varepsilon(n) = \exp(-C_1 h(n))$,
    shift inaccuracy $\eta = {C_2}h(n)$
    and failure probability $\delta(n) < \exp(-n)$, 
    $A$ can be transformed into an algorithm $A^\prime$ that solves the average-case trace reconstruction problem with 
    probability $1 - o_n (1)$, 
    time complexity $(\max_{n^\prime \leq C_2 \log n} \{ t(n^\prime)\} + n^{o(1)}) n$ 
    and a sample complexity of $\exp(C_2 h(C_2 \log n))$.
\end{theorem}

\begin{remark}
    Note that the assumption that $\log(n) \leq h(n) \leq \sqrt{n}$ is not very restrictive, since we show an upper bound of $h(n) \leq \widetilde{O}(n^{1/5})$ and the lower bound by Chase \cite{chase2021new} implies that $h(n) \geq 3/2 \log n$.
\end{remark}

This theorem also can also have an additional theoretical importance, when compared to Lemma 10 of \cite{mcgregor2014trace}.
In this lemma, McGregor et al show that if $f(n)$ traces are required for the worst-case trace reconstruction, then $f(\log n)$ traces are required for the average-case trace reconstruction.
In some sense \thref{thm:reduction} indicates that McGregor et al's theorem may be nearly tight, since we show that if $f(n)$ traces suffice for the shifted trace reconstruction then $\textup{poly}(f(\log n))$ traces suffice for the average-case.

Finally, we generalize Chase's upper bound (originally covering only deletion channels \cite{chase2021separating}) to the shifted trace reconstruction (\thref{thm:worst_case}) and by extension to the insertion-deletion channel.
Additionally, when the deletion probability is below $1/2$, we show that the worst-case trace reconstruction can be performed in $\exp(\widetilde{O}(n^{4/5}))$ time.
Formally, we show that:

\begin{theorem}
\thlabel{thm:worst_case}
    For any constant $q, q^\prime \in [0, 1)$ 
    and for any $C_2 > 0$, 
    there exist some $C_3 > 0$, $h(n) = C_3 n^{1/5} \log^7 n$ 
    and an algorithm $A$ 
    that solves the shifted trace reconstruction problem with 
    a shift inaccuracy of $ {C_2} h(n) $, 
    a sample complexity of $\exp(h(n))$,
    and a false sample rate of $\exp(-h(n))$.
    
    Furthermore, when $q < 1/2$, the algorithm $A$ runs in $\exp(O(n^{4/5} \log n))$ time and if $q \geq 1/2$, then $A$ runs in $\exp(O(n))$ time.
\end{theorem}

\begin{remark}
\label{rem:harder}
    The shifted trace reconstruction is at least as hard as the worst-case trace reconstruction, since one can simply set $S = 0$ and pad both the original message $\bfx$ with $0$ bits and the traces with traces of $0$ strings.
\end{remark}

Theorems \ref{thm:reduction} and \ref{thm:worst_case} directly imply \thref{thm:main_result}.
Furthermore, combining \thref{thm:worst_case} with Remark \ref{rem:harder} shows that $\exp(\widetilde{O}(n^{1/5}))$ samples suffice for the worst-case trace reconstruction problem for insertion-deletion channels as well.

Much of our paper will be devoted to generalizing and combining the results of Holden et al \cite{holden2020subpolynomial} and Chase \cite{chase2021separating}.
For the sake of brevity, we will henceforth refer to these papers as the HPPZ and the Chase constructions respectively.
We cite their main results in Section \ref{subsec:hppz_chase_results}.

\subsection{Organization}

In Section \ref{sec:prelim}, we give a precise definition of the trace reconstruction and present some notation.
We adapt the alignment technique of HPPZ to prove the general reduction in Section \ref{sec:reduction}.
Sections \ref{sec:fourier} and \ref{sec:analysis} contain the heart of our analysis, where we convert the shifted trace reconstruction problem into a complex analysis one (\ref{sec:fourier}) and use complex analysis techniques to solve it (\ref{sec:analysis}).
Finally, Appendices \ref{app:alignment}, \ref{app:interpolation} and \ref{app:proof_of_chase} contain some of the more technical proofs required for Sections \ref{sec:reduction}, \ref{sec:fourier} and \ref{sec:analysis} (respectively).

\section{Preliminaries}
\label{sec:prelim}

Let $\N = \{0, 1, \ldots \}$ and let $\mathcal{S} = \{0, 1\} ^ \N$ denote the set of infinite sequences of zeros and ones.
We denote elements $\bfx \in \mathcal{S}$ by $\bfx = (x_0, x_1 , \ldots)$, and denote by $\bfx(i:j) = \bfx([i, j]) = (x_i , \ldots x_j)$.

Fix any deletion probability $q \in \left[0, 1\right)$ and any insertion probability $q^\prime \in \left[0, 1\right)$.
Let $p = 1-q$ and $p^\prime = 1 - q^\prime$, and let $S$ be some $\N$ valued distribution.

Each trace $\widetilde{\bfx}$ is constructed from the original message $\bfx$ using the following procedure \cite{holden2020subpolynomial, cheraghchi2020overview}:

First, if we are producing a {\em shifted trace}, a shift $s$ is drawn independently from the distribution $S$, and the string $\bfx$ is replaced with its shift $\theta^s (\bfx) = \bfx(s:)$ (for other versions of the problem, we set $s = 0$ or skip this step).
Then, for each $j\in \N$, a random variable $G_j \geq 0$ is drawn from an independent geometric distribution $\Pr[G_j = \nu] = (q^\prime) ^ {\nu} p^\prime$.
For each $j$, $G_j$ independent and uniformly distributed bits are inserted before the $j$th bit of $\bfx$.
Finally, each bit of the resulting message is deleted with probability $q$.

We will often separate the randomness of the channel which we will denote by $\omega$ from the randomness generating the original message.

For any index $j$ we denote by $D_j$ the event that the $j$th bit $\bfx$ was not deleted by the channel.
Whenever $D_j$ occurs, we will denote by $f(j)$ the position within the trace to which this bit was sent (i.e. the number of bits either from $\bfx$ or from the insertions before this index that were not deleted).
If the $j$th bit was deleted, we define $f(j)$ to be $f(j^\prime)$ for the smallest $j^\prime \geq j$ for which $D_{j^\prime}$ holds.

In other words, $f(j)$ is defined to be the index in the trace $\widetilde{\bfx}$ that corresponds to the $\geq j$th index in the original message $\bfx$.
We will similarly denote by $I_j$ the event that the $j$th bit of the trace originated from the original message and not an insertion, and for such indices we will denote by $g(j)$ the index of the original message from which they are generated, and for indices which were the result of an insertion, we define $g(j)$ by the next non-inserted index in the trace.

We will define the misalignment between the $k$th bit of the input message and the $k^\prime$th bit of the trace to be:
\[
d(k, k^\prime) = \max\left\{\lvert f(k) - k ^ \prime\rvert, \lvert g(k ^ \prime) - k \rvert \right\}
\]

\begin{definition} [Shifted Trace Reconstruction Problem]
\label{def:shifted_trp}
    A {\em shifted trace reconstruction problem} has the following parameters:
    \begin{itemize}
        \item Channel parameters $q, q^\prime$
        \item Shift inaccuracy $\eta(n)$
        \item Sample complexity $\sigma(n)$
        \item False sample rate $\varepsilon(n)$
    \end{itemize}
    
    It is defined as the problem of reconstructing the $n+1$th bit $x_n$ of any $\bfx \in \{0, 1\}^\N$, given $\bfx(0:n-1)$, $\sigma(n)$ samples $\widetilde{\bfx}$, each of which is independently with probability $\varepsilon(n)$ selected from some unknown (potentially adversarial) distribution or with probability $1 - \varepsilon(n)$ a trace of $\theta^S(\bfx)$, where $S$ is some $\N$ valued random variable such that $\supp(S) \subseteq [a, a + \eta(n)]\subseteq[0, n-1]$ for some $a$.
\end{definition}

\subsection{The Main Results of HPPZ and Chase}
\label{subsec:hppz_chase_results}
\begin{theorem} [Theorem 1 \cite{holden2020subpolynomial}]
For $n \in \N$, let $\bfx \in \{0, 1\}^n$ be a bit-string where the bits are chosen uniformly and independently at random. Given $q, q^\prime \in [0, 1)$, there exists $M > 0$ such that for all $n$, we can reconstruct $\bfx$ with probability $1 - o_n(1)$ using $\lceil \exp(M \log^{1/3} n) \rceil$ traces from the insertion-deletion channel with parameters $q, q^\prime$.
Moreover, this can be done in $n^{1+o(1)}$ time.
\end{theorem}

\begin{theorem}[Theorem 2 \cite{chase2021separating}]
For any deletion probability $q \in (0, 1)$ and any $\delta > 0$, there exists $C > 0$, such that any unknown string $\bfx \in \{0, 1\}^n$ can be reconstructed with probability $1 - \delta$ from $\exp(C n^{1/5} \log^5 n)$ independent traces of $\bfx$.
\end{theorem}

\section{Proof of \texorpdfstring{\thref{thm:reduction}}{Theorem \ref{thm:reduction}}}
\label{sec:reduction}

In this section, we will prove the reduction from the average trace reconstruction to the shifted trace reconstruction.
The construction presented in this section requires only an adaptation of the HPPZ's methods, and our main contribution here is that we show that it can be used as a general reduction.
We will cite the main relevant theorems, but where no changes to the proof are necessary, we will refer readers to HPPZ's paper \cite{holden2020subpolynomial}.

Similar to HPPZ's algorithm, our reduction will be comprised of three main ingredients:

\begin{itemize}
    \item A Boolean test $T(\bfw, \widetilde{\bfw})$ on pairs of bit-strings $(\bfw, \widetilde{\bfw})$ that returns $1$ if $\widetilde{\bfw}$ is a plausible match for the output of applying the channel to $\bfw$.
    \item A two step alignment procedure comprised of a coarse and a fine alignment each of which uses the test $T$ to obtain an estimate $\tau^k$ for the positions in sufficiently many of the traces corresponding to the $k$th bit of the original message $\bfx$.
    \item A bit recovery procedure based on the target of our reduction to produce an estimate of any bit of $\bfx$ from sufficiently many aligned traces.
\end{itemize}

Finally, similar to HPPZ, throughout this section we will perform our analysis when $q = q^\prime$, but all of these results can be similarly generalized for any values of $q, q^\prime \in [0, 1)$.

\subsection{The Boolean Test}
\label{subsec:boolean_test}
The Boolean test $T$, as defined in \cite{holden2020subpolynomial}, is designed to answer whether a string $\widetilde{\bfw}$ is likely to have originated from a trace of some string $\bfw$.
It is controlled by two parameters $\ell, \lambda < \sqrt{\ell}$ and works by separating both strings into subsegments of length $\lambda$ each and comparing the average of each message on each segment.
If in sufficiently many segments the averages of the messages are either both above $1/2$ or both below $1/2$, the test returns $1$.
In other words, for any constant $c > 0$:

\[
    T^c _{\ell, \lambda}\left(\bfw, \widetilde{\bfw}\right) = 
    \begin{cases}
        1 & \text{if } \sum_{1 \leq i \leq \frac{\ell}{\lambda}} \text{sign}\left( s_i \cdot \widetilde{s}_i\right) \geq c \frac{\ell}{\lambda} \\
        0 & \text{otherwise}
    \end{cases}
\]
Where $s_i \sum_{(i-1)\lambda < j \leq i \lambda} (2w_i - 1) $.

HPPZ use this test with two sets of parameters. 
In both cases $c = \Theta(1)$ and up to some constant factors, $\ell, \lambda$ are either $\Theta(\log^{5/3} {n}), \Theta(\log ^ {2/3} {n})$ (for the ``coarse" alignment) or $\Theta(\log^{1/3} {n}), \Theta(1)$ (for the ``fine" alignment) respectively.

In the general case, we will set
\[
\ell = \Theta(\log^2 (n) / h(\Theta(\log(n)))); \;\;\;\;\; \lambda = \Theta(\log(n) / h(\Theta(\log(n))))
\]
for the coarse alignment and 
\[
\ell = \Theta(h(\log(n))); \;\;\;\;\; \lambda = \Theta(1)
\]
for the fine alignment, where $h(n)$ is the logarithm of the number of traces needed for the target of the reduction.
When $h(n)$ is $n^{1/3}$ (as in HPPZ's paper), we get the same parameters, and for our main result we will apply this theorem with $h(n) = \Theta(n^{1/5} \log^7 (n))$.

Roughly speaking, the boolean test should maintain two behaviours:
\begin{itemize}
    \item If $\widetilde{\bfw}$ is not a trace of $\bfw$, the probability that $T$ will return $1$ (called a {\em spurious match}) should be at most $\exp(-\Omega(\ell / \lambda))$. 
    \item If $\widetilde{\bfw}$ is a trace of $\bfw$, the probability that $T$ will pass (called a {\em true match}) will be at least $\exp(-O(\ell / \lambda^2))$. 
\end{itemize}

Under these constraints, the probability of a true match may be very small, but when $\lambda$ is sufficiently large, it will be much higher than the probability of a spurious match.
When conditioning on a match, it will most likely be a true match.

In order to formalize this, HPPZ define a condition for the robustness of this test and a notion for the error in an alignment.
They prove that almost any string is robust and that for robust strings there is a sufficiently high probability to have a true match.
Furthermore, they prove that for robust and ``mismatched" strings the probability of a spurious match is sufficiently low.

More formally the robustness is defined as below, where $\theta$ is a constant as defined in HPPZ:

\begin{definition} [Definition 3 of \cite{holden2020subpolynomial} - Robustness]
\label{def:robustness}
    Let $u_2 = u_1 + \lambda$ be two indices in the string $\bfx$. We define the {\em robust bias} of $\bfx(u_1 + 1: u_2)$ to be
    \[
        \lambda^{-\frac{1}{2}} 
        \sum_{\substack{
            t_1, t_2 \in \N \\ 
            \lvert t_1 - u_1\rvert < \lambda / 100 \\
            \lvert t_2 - u_2 \rvert < \lambda / 100
        } }
        \lvert
            \sum_{j = t_1} ^ {t_2} (2x_j - 1)
        \rvert
    \]
    We will say that $\bfx(u_1 + 1: u_2)$ has a {\em clear robust bias} if its robust bias is at least $1$.

    Let $\bfw$ be some string of length $\ell$ (usually a substring of $\bfx$). We will say that $\bfw$ has a {\em clear robust bias at scale $\lambda$}, if when separating $\bfw$ into blocks of length $\lambda$ (i.e. viewing the substrings $\bfw(u_i+1, u_{i+1})$ where $u_i = i \lambda$), at least $\theta$ fraction of them have a clear robust bias.
\end{definition}

\begin{definition} [Definition 7 of \cite{holden2020subpolynomial} - Mismatched strings]
\label{def:mismatch}
    Let $\bfw = \bfx(a+1:a+\ell)$ be a substring of the input and let $\widetilde{\bfw} = \widetilde{\bfx}(b+1:b+\ell)$ be a substring of the same length taken from the trace. We say that $\bfw$ and $\widetilde{\bfw}$ are {\em $s$-mismatched} if for any $0 \leq i \leq \ell$, it holds that $d(a+i, b+i) \geq s$.
\end{definition}

Heuristically, a string $w$ which is robust in the sense of Definition \ref{def:robustness} should be matched with any one of its traces with probability $\exp(-O(\ell / \lambda^2))$.
This is because so long as the number of deletions doesn't drift by more than $O(\lambda)$ from the number of insertions, we expect the proportion of $0$s to $1$s in each subsegment of the trace to be strongly correlated with the proportion of $0$s to $1$s in the parallel subsegment of the input message.
HPPZ formalize this statement, proving the following lemmas:

\begin{lemma} [Lemma 4 of \cite{holden2020subpolynomial} - Most strings are robust]
\thlabel{lem:most_robust}
    Let $\bfw$ be a random string of length $\ell$. Then $\bfw$ has a clear robust bias at scale $\lambda$ w.p. $1 - \exp(-\Omega(\ell/\lambda))$.
\end{lemma}

\begin{lemma} [Lemma 5 of \cite{holden2020subpolynomial} - Robust strings have a true match with sufficiently high probability]
\thlabel{lem:true_match}
    For any constant $q, q^\prime \in [0, 1)$, 
    there exists some constant $c > 0$, so that for any 
    $\bfw \in \{0, 1\}^\ell$ 
    which has a clear robust bias at scale $\lambda$, 
    if $\widetilde{\bfw}$ is a trace of $\bfw$, 
    then the test $T_{\ell, \lambda}^c (\bfw, \widetilde{\bfw}(0:\ell-1))$ 
    will pass with probability at least $\exp(-O(\ell/\lambda^2))$.
\end{lemma}

HPPZ then prove the first property of the test, namely that it has very few spurious matches.

\begin{lemma} [Lemma 8 of \cite{holden2020subpolynomial} - False positives are rare]
\thlabel{lem:false_positives}
    Let $\bfx$ be a random string and suppose we sample the trace $\widetilde{\bfx}$ from the insertion-deletion channel.

    Consider two length-$\ell$ substrings $\bfw = \bfx(a+1:a+\ell)$ and $\widetilde{\bfw} = \widetilde{\bfx}(b+1:b+\ell)$.
    For any realisation $\omega_0$ of the randomness of the channel (which determines the map of insertions and deletions), such that $\bfw$ and $\widetilde{\bfw}$ are $\lambda$-mismatched, we have that:
    \[
    \Pr_{\bfx} \left[T_{\ell, \lambda} \left(\bfw, \widetilde{\bfw}\right) \mid \omega = \omega_0\right] \leq \exp\left(-\Omega\left( \ell / \lambda \right) \right)
    \]
\end{lemma}

\begin{definition} [Definition 9 of \cite{holden2020subpolynomial} - spurious matches]
    Let $\ell, \lambda$ be given positive integers.
    Let $\bfx$ be an input string, let $I$ be an interval of length $\ell$ and write $\bfw = \bfx(I)$.
    Let $J$ be another interval (usually containing $I$).

    Let $\widetilde{\bfx}$ be a trace of $\bfx$ through the insertion-deletion channel.
    We say that an {\em $(I, J)$-spurious match} occurs 
    if for some substring of the trace $\widetilde{\bfw} = \widetilde{\bfx}(i_1 : i_2)$, 
    such that $g([i_1, i_2]) \subseteq J$ (i.e. whose bits originated from within the interval $J$), 
    we have $T_{\ell, \lambda} (\bfw, \widetilde{\bfw}) = 1$, but $\bfw$ and $\widetilde{\bfw}$ are $\lambda$-mismatched.

    We will denote the event that an $(I,J)$-spurious match occurs by $\mathcal{Q}_{\ell, \lambda} (I, J)$.
\end{definition}

\begin{lemma} [Lemma 10 of \cite{holden2020subpolynomial} - spurious matches are rare within an interval]
\thlabel{lem:within_interval}
    Let $\ell$ and $\lambda$ be given positive integers.
    Let $I$ be an interval of length $\ell$ and let $J\supseteq I$ be an interval containing $I$.
    Suppose we have an input string $\bfx$ all of whose bits are determined except those in $J$, which are drawn i.i.d. uniformly.
    Then:
    \[
    \Pr\left(\mathcal{Q}_{\ell, \lambda} (I, J)\right) \leq \lvert J \rvert e^{-\Omega(\ell / \lambda)} + e^{-\Omega\left(\lvert J \rvert \right)}
    \]
\end{lemma}

\begin{lemma} [Lemma 11 of \cite{holden2020subpolynomial} - spurious matches are rare between different intervals]
\thlabel{lem:cross_interval}
    Let $\ell$ and $\lambda$ be given positive integers.
    Let $I$ be an interval of length $\ell$ and let $J$ be a disjoint interval whose distance from $I$ is at least $\lvert J \rvert$.
    Suppose we have an input string $\bfx$ all of whose bits are determined except those in $I$, which are drawn i.i.d. uniformly.
    Then:
    \[
    \Pr\left(\mathcal{Q}_{\ell, \lambda} (I, J)\right) \leq \lvert J \rvert e^{-\Omega(\ell / \lambda)} + e^{-\Omega\left(\lvert J \rvert \right)}
    \]
\end{lemma}

We will use these lemmas exactly as proven by HPPZ, so for the sake of brevity, we will not repeat the proofs.

\subsection{Coarse and Fine Alignments}
\label{subsec:alignment}

The next step of both our reduction and HPPZ's algorithm is to perform coarse and fine alignments.
In this portion of their construction, HPPZ set their parameters specifically for their $\exp(O(\log^{1/3} n))$ sample algorithm, so it will require some minor changes for our case.

In this section, we will define the properties we want the string $\bfx$ to have in order for each step in our alignment procedure to succeed.
In Appendix \ref{app:alignment} we will prove that if $\bfx$ is randomly chosen, it maintains these properties with high probability.

Let be $C$ a sufficiently large constant.
We define the parameters for the coarse and fine alignments to be:

\[
\ell_c = C \frac{\log^2 n}{h(C \log n)};\;\;\;\;\;\; \lambda_c = C^{1/2} \frac{\log n}{h(C \log n)}
\]

\[
\ell_f = C^{2/3} h(C \log n);\;\;\;\;\;\; \lambda_f = C^{1/12}
\]

\subsubsection{Coarsely Well-Behaved Strings}

\begin{definition}[Coarsely well-behaved strings]
\label{def:coarsely}
    Let $\bfx$ be a string of length $n$ and let $\ell_c, \lambda_c$ be as defined above.
    We say that $\bfx$ is {\em coarsely well-behaved}, if for each interval $I \subseteq [0, n]$ of length $\ell_c$, it holds that $\bfx(I)$ has robust bias at scale $\lambda_c$ and 
    \[
        \Pr_{\omega}\left[ \mcQ_{\ell_c, \lambda_c}(I, [0, n]) \right] \leq n^{-2}
    \]
    (where the probability is taken over the noise of the channel)
\end{definition}

\begin{lemma}
\thlabel{lem:coarsely_well_behaved}
    Let $\bfx$ be a random string of length $n$.
    Then, $\bfx$ is coarsely well-behaved with probability at least $1 - n^{-2}$.
\end{lemma}

\subsubsection{Finely Well-Behaved Strings}

Recall Lemmas \ref{lem:within_interval} and \ref{lem:cross_interval}. Let $c_0$ be such that the constant factors in the $\Omega(\ell / \lambda)$ and $\Omega(|J|)$ were at least $10 c_0$.

\begin{definition}[Finely well-behaved strings]
\label{def:finely}
    Let $\bfx$ be a string of length $n$, and let $\ell = \ell_f, \lambda = \lambda_f$ as defined above.
    We say that $\bfx$ is {\em finely well-behaved} if for each interval $J = [a, a+C\log n] \subseteq [0, n]$ of length $C \log n$, there exists a sub-interval $I \subseteq [a+ 1/3 C \log n, a + 2/3 C \log n]$ of length $\ell$, such that $\bfx(I)$ exhibits robust bias at scale $\lambda$ and $\Pr_\omega[\mcQ_{\ell, \lambda} (I, J)] \leq \exp(-c_0 C^{7/12} h(C \log n))$.
\end{definition}

\begin{lemma}
\thlabel{lem:finely_well_behaved}
    Let $\bfx$ be a random string of length $n$.
    Then, $\bfx$ is finely well-behaved with probability at least $1 - n^{-2}$.
\end{lemma}

\subsection{Using the Oracle}

In Section \ref{subsec:boolean_test}, we showed that the boolean test $T$ has several very nice properties when the input string $\bfx$ is well-behaved, and in Section \ref{subsec:alignment}, we proved that almost all strings are well-behaved.
For the rest of this section, we will denote by $\bad$ the case where $\bfx$ is not well-behaved (coarsely or finely), and by $\good$ the case where the $\bfx$ is well-behaved.

For any well-behaved string $\bfx \in \good$ and any integer $k \in [\ell_c + C \log n, n]$, we set $a_1 = k - \ell_c - C\log n$ and $a_2$ to be such that
\[
I = [a_2, a_2 + \ell_f] \subseteq J = [k - 2/3 C \log n, k - 1/3 C \log n]
\]
is the interval promised by our assumption that $\bfx$ is finely well-behaved.

For any trace $\widetilde{\bfx}$, we set our {\em coarse alignment} $\tau^k_1$ to be the first integer $b$ for which
\[
T_{\ell_c, \lambda_c}(\bfx([a_1, a_1+\ell_c]), \widetilde{\bfx}([b, b+\ell_c])) = 1
\]
or $\infty$ if no such $b$ exists.

For any trace $\widetilde{\bfx}$, if $\tau^k_1 < \infty$, we set our {\em fine alignment} $\tau^k_2$ to be the first integer
\[
b \in [\tau^k_1 -\ell_c, \tau^k_1 +2\ell_c + C \log n]
\]
such that:
\[
T(\bfx([a_2, a_2 + \ell_f]), \widetilde{\bfx}([b, b+\ell_f])) = 1
\]
If $\tau^k_1 = \infty$ or no such $b$ exists, we set $\tau^k_2 = \infty$.

From the definitions of Section \ref{subsec:alignment} and the lemmas of Section \ref{subsec:boolean_test} it will be easy to show that the following properties hold:

\begin{lemma}
\thlabel{lem:taus}
    Let $\bfx \in \good$ be a well-behaved string and let $k \in \{\ell_c + C \log n, \ldots, n\}$ be an integer.
    Then for $a_1, a_2, \tau^k_1, \tau^k_2$ as defined above, the following properties hold:
    \begin{itemize}
        \item $\Pr \left[\tau^k_1 < \infty\right] > \exp (-c_1 C^{1/2} h(C \log n))$
        \item $\Pr \left[\tau^k_1 < \infty \wedge d(k, \tau^k_1) > \ell_c\right] < n^{-2}$
        \item $\Pr \left[\tau^k_2 < \infty \mid \tau^k_1 < \infty\right] \geq \exp(-c_2 C^{1/2} h(C \log n))$
        \item $\Pr \left[\tau^k_2 < \infty \wedge d(k, \tau^k_2) > \ell_f \mid \tau^k_1 < \infty\right] < \exp(-c_3 C^{7/12} h(C \log n))$
    \end{itemize}
    Where the probabilities are taken over the randomness of the channel and $c_1, c_2, c_3, c_4 > 0$ are positive constants that may depend on $q, q^\prime$ but not on $C$ or $n$ and originate from the $\Omega(\cdot)$s and $O(\cdot)$s of the previous sections.
\end{lemma}

The first two parts of this claim follow directly from our definition of coarsely well-behaved strings (Definition \ref{def:coarsely}) and \thref{lem:true_match}.
The rest of it follows directly from our definition of finely well-behaved strings (Definition \ref{def:finely}) and the same lemma.

In order to prove the main claim of our reduction we will need one more lemma which we will prove in the next subsection:

\begin{lemma} [$\tau^k_1, \tau^k_2$ can be computed efficiently]
\thlabel{lem:efficient}
    There is an algorithm $A_{\textup{align}}$ such that, 
    for any $\bfx \in \good, k \in \{\ell_c + C \log n, \ldots, n\}$ and 
    any trace $\widetilde{\bfx}$ of $\bfx$ through the channel, 
    given $k, \bfx(0:k), (\tau_1^1, \ldots , \tau_1^{k-1}), (\tau_2^1, \ldots, \tau_2 ^ k)$, 
    $A_{\textup{align}}$ computes $\tau^k_1, \tau^k_2, a_2 $ in time $n^{o(1)}$, 
    with probability $\geq 1 - n^{-2}$.
\end{lemma}

Before proving \thref{lem:efficient}, we will show that the main theorem of our reduction (\thref{thm:reduction}) follows immediately from it.

\begin{proof}[Proof of \thref{thm:reduction}]
    Let $q, q^\prime \in [0, 1)$ be the parameters of the channel, and set $C$ to be a sufficiently large constant\footnote{This choice of $C$ it not meant to be tight.}:
    \[
    C = (100\max \{ 1, 1/c_0, C_1, 1/C_1, 1/c_1, 1/c_2, 1/c_3, c_0, c_1, c_2, c_3, c_4\}) ^ {100}
    \]
    (where $c_0, c_1, c_2, c_3$ are the constants from the Definition \ref{def:finely} and \thref{lem:efficient} and $C_1$ is the constant from \thref{thm:reduction}).

    We set the constant $C_2$ of \thref{thm:reduction} to be equal to $C$.

    We will prove that given the first $k \geq \ell_c + C\log n$ bits of $\bfx$, we can reconstruct the rest.
    We can work under this assumption, by adding $\ell_c + C\log n$ virtual $0$ bits to the start of $\bfx$ and adding a trace of $0^k$ to the beginning of each of the traces $\widetilde{\bfx}$ before the reconstruction.

    Given the first $k$ bits of $\bfx$, we will show that we can reconstruct the $k+1$th bit of $\bfx$ and from there, we can continue this process iteratively.
    Using the alignment algorithm from \thref{lem:efficient}, we compute $\tau_1 ^ k$ and $\tau_2 ^k$ of each of the traces $\widetilde{\bfx}$.

    Given $a_2, \tau_2 ^k$, we run the shifted trace reconstruction algorithm $A$ with parameters $n^\prime, n^\prime - 1$, where $n^\prime = k - a_2 \in [1/3 C \log n, 2/3 C \log n]$, on the set:
    \[
    \mathcal{X} = \left\{\widetilde{\bfx}(\tau_2 ^k:) \mid \substack{
        \widetilde{\bfx}\text{ is a sample}\\
        \tau_2^k(\widetilde{\bfx}) < \infty} \right\}
    \]

    The first and third claims of \thref{lem:taus}, mean that for each of our $N = \exp(C h(C \log n))$ traces, it will have a finite $\tau_2 ^ k$, with probability at least
    \[
    \exp(-C^{1/2} (c_1 + c_2) h(C \log n)) \geq \exp(-1/3 C h(C \log n)).
    \]
    Therefore, by Hoeffding's inequality, the probability that we will have at least 
    \[
    1/2 \exp(2/3 C h(C \log n)) > \exp(h(2/3 C \log n)) \geq \exp(h(k - a_2)) 
    \]
    traces for which $\tau_2 ^k < \infty$ is at least 
    \[
    1 - \exp(-\Omega(C h(C \log n))) = 1 - n^{-\omega(1)}
    \]

    \thref{lem:taus} gives us that the probability that any message for which $\tau_2^k < \infty$ is the result of a spurious match is at most 
    \[
    \varepsilon(n) \leq \exp(-(C^{7/12} c_3 - C^{1/2} (c_1 + c_2)) h(C\log n)) \leq \exp (-C_1 h(C \log n))
    \]
    We will make no assumptions about the strings that came from spurious alignments.
    By definition, any $\tau_2 ^ k < \infty$ that was not the result of a spurious match, had distance $d(a_2, \tau_2^k)\leq \eta = \lambda_f$.

    Therefore, the samples in $\mathcal{X}$, constitute a shifted trace reconstruction problem.
    Because we assume $A$ solves the shifted trace reconstruction problem with probability $1 - \exp(-n^\prime)$, and we will be applying $A$ on messages of length $n^\prime \geq 1/3 C \log n$, it will succeed with probability $\geq 1 - \exp(-n^\prime) > 1 - n^{-2}$.

    Taking a union bound on the values of $k$, we see that $A$ will succeed in resolving the value of $x_k$ for all $k$, w.p. $\geq 1 - 1/n$.

\end{proof}

\subsection{Time Complexity (Proof of \texorpdfstring{\thref{lem:efficient}}{Lemma \ref{lem:efficient}})}

All that is left in order to prove the reduction (\thref{thm:reduction}) is to show that $\tau_1, \tau_2, a_2$ can be computed efficiently and with a high success rate.

For any trace $\widetilde{\bfx}$, given $\tau_1$ and $a_2$, it is easy to compute $\tau_2$ by setting $I = [a_2, a_2 + \ell_f]$ and for each integer $ \tau_1 - \ell_c \leq b \leq \tau_1 + 2 \ell_c + C \log n$, we perform the test $T(\bfx(I), \widetilde{\bfx}([b, b+\ell_f]))$, outputting the first $b$ for which it returns $1$.
This requires $O(\ell_c)$ iterations of a test that takes $O(\ell_f)$ time, for a total of $n^{o(1)}$.

$a_2$ is defined as being some index in $[k - 2/3 C \log n, k - 1/3 C \log n]$ for which two properties hold (whose existence is promised by our assumption that $\bfx$ is finely well-behaved).
First, the segment $\bfx([a_2, a_2 + \ell_f])$ exhibits robust bias at scale $\lambda_f$, and this is easy to check in $\text{poly}(\ell_f)$ time.

The second property is that we want the probability of having a spurious match between $I$ as defined above and any subinterval of $J = [k - C \log n, k]$ to be lower than some $\exp(-\Omega(h(C \log n)))$.
In order to check if $a_2$ maintains this property, we will generate a sufficiently large (but still $\exp(O(h(C \log n))) = n^{o(1)}$) number of traces of $\bfx(J)$ and for each possible value of $a_2$, we will count the number of sub-intervals of $J$ for which $I$ has a spurious match.
Using standard probability bounds, we can show that this process will allow us to approximate $\mcQ_{\ell_f, \lambda_f} (I, J)$ to a sufficiently high degree of accuracy with a $n^{-\omega(1)}$ failure rate and $n^{o(1)}$ complexity.

This leaves us with the task of computing $\tau_1 ^ k$ efficiently.
To do so, we search for the last finite $\tau_1 ^ j < \infty$ of this trace.
From Hoeffding's inequality, it is easy to see that with probability $\geq 1 - n^{-\omega(1)}$, we will have $j \geq k - 100 \ell_c \log^2 n$ and from the Chernoff bound that with a similarly high probability, if $\tau_1^k$ and $\tau_1^j$ both resulted from real matches (i.e. not spurious ones) then $\tau_1 ^ k \leq \log^2 n (k - j) + \tau_1 ^j$.

Combining both of these, it suffices to check $n^{o(1)}$ options for $\tau_1^k$.
Since each of these tests takes $n^{o(1)}$ time, this step also has a complexity of $n^o(1)$, proving \thref{lem:efficient}.


\section{Conversion to Complex Analysis}
\label{sec:fourier}

Like many other results regarding the trace reconstruction problem (such as \cite{nazarov2017trace, de2017optimal, peres2017average, holden2020subpolynomial, chase2021separating}), our proof of \thref{thm:worst_case} will rely on a complex analysis based on the results of Borwein and Erd\'{e}lyi's seminal research on Littlewood polynomials \cite{borwein1997littlewood}.
These analyses are typically based on proving that some polynomial related to the input message is equal to the average of a property of the traces.

In this section, we will adapt the ``non-linear" complex analysis in Chase's construction (\thref{prop:chase_6_2}) to  insertion-deletion channels with random shifts using a generalization of the analysis shown by HPPZ (\thref{lem:peres_22}). 

\begin{proposition} [Proposition 6.2 of \cite{chase2021separating}]
\label{prop:chase_6_2}
For any $x \in \{0,1\}^n, l \ge 1, w \in \{0,1\}^l$, and $z_0,\dots,z_{l-1} \in \C$, we have $$\E_x\left[p^{-l} \sum_{j_0 < \dots < j_{l-1}} \left(\prod_{i=0}^{l-1} 1_{U_{j_i} = w_i}\right)\left(\frac{z_0-q}{p}\right)^{j_0}\left(\prod_{i=1}^{l-1} \left(\frac{z_i-q}{p}\right)^{j_i-j_{i-1}-1}\right)\right]$$ $$ = \sum_{k_0 < \dots < k_{l-1}} \left(\prod_{i=0}^{l-1} 1_{x_{k_i}=w_i}\right) z_0^{k_0}\left(\prod_{i=1}^{l-1} z_i^{k_i-k_{i-1}-1}\right).$$
\end{proposition}

\begin{lemma} [Lemma 22 of \cite{holden2020subpolynomial}]
\thlabel{lem:peres_22}
Let $S$ be a bounded $\N$-valued random variable. Let
  $\bfa=(a_0,a_1,\dots)\in[-1,1]^{\N}$, and let $\wt{\bfa}$ be the output from
  the insertion-deletion channel with deletion (resp.\ insertion)
  probability $q$ (resp.\ $q'$), applied to the randomly shifted
  string $\theta^{S}(\bfa)$. Let $\phi_1(w)=pw+q$, $\phi_2(w) =
  \frac{p'w}{1-q' w}$, and $\sigma(s)=\Pr[S=s]$ for $s\in\N$. Define
  \[ P(z) := \sum_{s=0}^{d} \sigma(s) z^s, \qquad Q(z):=\sum_{j= 0}^\infty a_j z^j. \]
  Then, for any $|w|<1$,
  \begin{equation} \label{eq:gen-func}
    \E\left[ \sum_{j\geq 0} \wt{a}_j w^j \right] =     
    \frac{pp'}{1-q'\phi_1(w)}
    \cdot
     P\left(\frac{1}{\phi_2\circ\phi_1(w)} \right ) \cdot Q(\phi_2\circ\phi_1(w)).
  \end{equation}
\end{lemma}

Ideally, we would want to directly combine these theorems.
However, \thref{lem:peres_22} works only when the entries inserted by the channel are taken from a centered distribution (i.e. they have a mean of $0$).
This is not problematic for HPPZ's analysis, since they only apply the lemma to the difference between the messages $a_j = x_j - y_j$ (where $\bfx$ and $\bfy$ are the messages between which one is trying to distinguish).

In contrast, for Chase's upper bound, one distinguishes between $f(\bfx) = \prod_i \left(x_{j_i} - w_i \right)$ and $f(\bfy)$ (for some $\bfw \in \left\{0, 1\right\}^l$).
Combining these techniques is not trivial, because $f(\bfx)-f(\bfy)$ cannot be written as a function of $\textbf{a} = \bfx - \bfy$.

Our goal in this section will be to show that we can overcome this problem.
In particular, we will prove the following theorem:

\begin{theorem}
\thlabel{thm:can_compute_g}
Let $h(n) = n^{1/5} \log^7 n$, $l \leq 2 n^{1/5} + 1$ and let $c_1 > c_2 > 0$ be sufficiently small constants.

Let $f(x_1, \ldots, x_l)$ be some function from $\{0, 1\}^l$ to $\D$ (in our case $f(x_1, \ldots, x_l) = \prod_i (x_i - w_i)$ for some $\bfw \in \{0, 1\}^l$).
We define the polynomial $g_\bfx ^ f$ to be the following:
\[
g_{\bfx} ^ f(\zeta_0 , \zeta_1, \ldots , \zeta_l) = \sum_{k_0 < \cdots < k_{l-1}} (-1)^{x_{k_0}} f(x_{k_1}, \ldots x_{k_l}) \zeta_0 ^{k_0 - 1} \zeta_1 ^{k_1 - k_0 - 1} \ldots \zeta_{l-1}^{k_{l}-k_{l-1} - 1}
\]

For any point $z_0, z_1, \ldots, z_l \in \C^{l+1}$, such that $z_0 = (1 - n^{-4/5} \log^{6} n) e^{i\theta}$, $\abs{\theta} \leq n^{2/5}$ and $z_1 = z_2 = \cdots = z_l \in [1-c_1, 1-c_2]$, given $\exp(h(n))$ traces of the shifted trace reconstruction problem with false-sample rate $\varepsilon(n) = \exp(-h(n))$ and shift inaccuracy $\eta = O(h(n))$, we can compute $g(z_0, z_1, \ldots, z_l)$ to within an additive error of $\pm \exp(-\Omega(h(n)))$, with probability $1 - \exp(-\omega(n))$ and in time $\exp(\wt{O}(h(n)))$.

Furthermore, when $q < 1/2$, this also holds for $z_1 = z_2 = \cdots = z_l \in [-c_1, c_1]$.
\end{theorem}

We separate the proof of \thref{thm:can_compute_g} into three parts.
In the first portion of the proof, we will show that the statement holds for the specific case where $f(x_1, \ldots, x_l) = \prod_i (-1)^{x_i}$ is a "simple character", even if the equality $z_1 = \cdots = z_l$ does not hold.

Then we will show that \thref{thm:can_compute_g} holds for any character $f(x_1, \ldots, x_l) = \prod_i \omega_i ^ {x_i} = \chi_{\omega}$ (for any $\omega \in \{-1, 1\}^l$).
This step does not trivially follow from the first one, because $f(x_1, \ldots, x_l)$ can now have "holes" - variables that do not affect its outcome, and this changes the polynomial $g_\bfx ^ f$.
We will show that when $z_1 = \cdots = z_l$, the $g_\bfx$ of a general character is a high-order differential of $g_\bfx$ of a smaller simple character, and that this numerical differentiation does not reduce the accuracy too much.

Finally, because the transformation from $f$ to the polynomial $g_\bfx ^f$ is linear, we can take any $f:\{0, 1\}^l \rightarrow \D$ and use its Fourier transformation over $\mathbb{F}_2 ^ l$ to show that:
\begin{equation}
\label{eq:fourier_separation}
    g_\bfx ^f(\zeta_0, \zeta_1, \ldots, \zeta_l) = g_\bfx ^{\sum_{\omega \in \{-1, 1\}^l} \hat{f}(\omega)\chi_\omega}(\zeta_0, \zeta_1, \ldots, \zeta_l) = \sum_{\omega \in \{-1, 1\}^l} \hat{f}(\omega) g_\bfx ^{\chi_\omega}(\zeta_0, \zeta_1, \ldots, \zeta_l).
\end{equation}

Combining equation \eqref{eq:fourier_separation} with the second step will yield \thref{thm:can_compute_g}.

\vs

The first step in our analysis will be the following lemma:

\begin{lemma} 
\thlabel{lem:complex_connection}
Let $S$ be a random variable, such that $\textup{supp}(S) \subseteq \{0, 1, \ldots, d\}$ for some finite $d$. 
Let $\bfa = \left(a_0, a_1, \ldots\right) \in \left\{0, 1\right\} ^ \N$, and let $\widetilde{\bfa}$ be the output from the insertion-deletion channel with deletion probability $q$ and insertion probability $q^\prime$ applied to the randomly shifted string $\theta^s (\bfa)$. 
Let $\phi_1 (z) = pz + q$, $\phi_2 (z) = \frac{p^\prime z}{1 - q^\prime z}$, $\phi \defeq \phi_2 \circ \phi_1$, $\Psi = \phi^{-1}$, $\overline{\phi}(z) = \frac{p p^\prime}{1 - q^\prime \phi_1(z)} = \frac{p\phi(z)}{\phi_1 (z)}$ and $\sigma(s) = \Pr[S = s]$ for $s \in \N$. 
Define
\[ P(z) \defeq \sum_{s=0} ^ {d} \sigma(s) z^s \]
Then:
\begin{equation}
    \begin{aligned}
        &\mathbb{E}\left[ \sum_{r_0 < \cdots < r_{l}} (-1)^{\widetilde{a}_{r_0}} \Psi(z_0)^{r_0 - 1} \left(\prod_{i=1}^{l} (-1) ^ {\widetilde{a}_{r_i}} \Psi(z_i) ^ {r_i - r_{i-1} - 1}\right) \right] \\
        &\;\;\;\;\;= \left(\prod_{0 \leq i \leq l} \overline{\phi}(\Psi(z_i))\right) P\left( \frac{1}{z_0} \right) \cdot \sum_{k_0 < k_1 < \cdots < k_l} (-1)^{a_{k_0}} z_0 ^ {k_0 - 1} \left(\prod_{i=1} ^ {l} (-1)^{a_{k_i}} z_i ^ {k_i - k_{i-1} - 1}\right)
    \end{aligned}
\end{equation}

\end{lemma}

This lemma is similar to a combination of \thref{lem:peres_22} and \thref{prop:chase_6_2}, but avoids the drift caused by the insertions by using a function that has $0$ mean when any of the bits is a random insertion.
\subsection{Proof of \texorpdfstring{\thref{lem:complex_connection}}{Lemma \ref{lem:complex_connection}}}

We consider both the shifted channel and the original one (i.e. with $S = 0$).
For the unshifted channel, let $A_{k, r}$ denote the event that the first $r$ bits of the trace were produced by the first $k$ bits of the message and that the $k$th bit of the message was not deleted (i.e. that after the shift, all the insertions before $\bfx_k$ and all the deletions, $r$ bits were left).
Similarly, let $A^\prime_{k, r}$ denote the same event for the shifted channel (when $S$ is as in the lemma).
We define:
\begin{equation}
    \begin{aligned}
        &\alpha_{k, r}       = \Pr[A_{k, r}]\\
        &\alpha\prime_{k, r} = \Pr[A\prime_{k, r}]
    \end{aligned}
\end{equation}

Since each of the bits of the message are translated to a geometric number of bits and then each is deleted or not independetly of the rest, we can use basic results on generating functions to produce a simple formula for $\alpha_{k,r}$ and $\alpha^\prime_{k,r}$:

\begin{equation}
\label{eq:generating_functions}
    \begin{aligned}
        \forall k\;\;\;\; &P_{\alpha_k}(\zeta) \defeq \sum_r \zeta^r \alpha_{k,r} = p \left(P_{\textup{Geometric}(q^\prime)}\left(P_{\textup{Bernoulli}(q)}(\zeta)\right)\right)^{k-1} P_{\textup{Geometric}(q^\prime) - 1}\left(P_{\textup{Bernoulli}(q)}(\zeta)\right) \zeta \\
        &\;\;\;\;\;\;\;\;\;\;\;\;\;\;\;\;\;\;\;\;  = \zeta \overline{\phi}(\zeta) \phi(\zeta) ^ {k-1}\\
        &P_{\alpha^\prime_k}(\zeta) \defeq \sum_r \zeta^r \alpha^\prime_{k,r} = \sum_{s} \Pr\left[S = s\right] \sum_r \zeta^r \alpha^\prime_{k-s,r} = \sum_s \sigma(s) \phi(\zeta) ^ {k-s -1} \zeta \overline{\phi}(\zeta) \\
        &\;\;\;\;\;\;\;\;\;\;\;\;\;\;\;\;\;\;\;\;  = \zeta P\left(\frac{1}{\phi(z)}\right) \overline{\phi}(\zeta) \phi(\zeta) ^ {k-1}
    \end{aligned}
\end{equation}

Setting $z = \phi(\zeta)$, and defining
\begin{equation*}
    \begin{aligned}
        &P_{\alpha_k-1}(\zeta) = \zeta^{-1} P_{\alpha_k}(\zeta)\\
        &P_{\alpha^\prime_k-1}(\zeta) = \zeta^{-1} P_{\alpha^\prime_k}(\zeta)
    \end{aligned}
\end{equation*}
gives us the formulas:

\begin{equation}
\label{eq:generating_functions2}
    \begin{aligned}
        &P_{\alpha_k - 1}(\Psi(z)) = \overline{\phi}(\Psi(z)) z ^ k\\
        &P_{\alpha^\prime_k - 1}(\Psi(z)) = \overline{\phi}(\Psi(z)) P(1/z) z ^ k
    \end{aligned}
\end{equation}

We denote by $B_r = \overline{\bigvee_{k} A_{k, r} }$ the event that the $r$th bit of the output was an insertion.
By our definition of the channel, conditioning on $B_r$, the $r$th bit of the $\widetilde{a}$ is a Bernoulli$\left(\frac{1}{2}\right)$ random variable independent of the rest of the problem.
Therefore, we have
\begin{equation}
    \label{eq:B_r}
    \mathbb{E}_{\textup{conditioned on }B_{r_i}} \left[\prod_{1 \leq j \leq l} (-1) ^ {\wt{a}_{r_j}}\right] = 0
\end{equation}

Let $k_0 < k_1 < \cdots < k_l$ be some indices in the input message and let $r_0 < r_1 < \cdots < r_l$ be some indices in the output message.
We consider the events
\begin{equation}
    \begin{aligned}
        &A_{\Vec{k}, \Vec{r}} \defeq \bigwedge_i A_{k_i, r_i}\\
        &A^\prime_{\Vec{k}, \Vec{r}} \defeq \bigwedge_i A^\prime_{k_i, r_i}
    \end{aligned}
\end{equation}

It is clear from our definition of the channel, that for all sequences of indices $r_0 < r_1 < \cdots < r_l$ and non-monotone sequences $k_0, k_1, \ldots, k_l$, the event $ A^\prime_{\Vec{k}, \Vec{r}}$ has probability $0$.
Furthermore, from the independence of the channel it is easy to see that:
\begin{equation}
    \label{eq:A_independence}
    \Pr\left[ A^\prime_{\Vec{k}, \Vec{r}} \right] = \alpha^\prime_{k_0, r_0} \prod_{1 \leq i \leq l} \alpha_{k_i - k_{i-1}, r_i - r_{i-1}} 
\end{equation}

We combine equations \eqref{eq:generating_functions2}, \eqref{eq:B_r} and \eqref{eq:A_independence} to get the desired result:

\begin{equation}
    \begin{aligned}
        &\mathbb{E}\left[ \sum_{r_0 < \cdots < r_{l}} (-1)^{\widetilde{a}_{r_0}} \Psi(z_0)^{r_0 - 1} \left(\prod_{i=1}^{l} (-1) ^ {\widetilde{a}_{r_i}} \Psi(z_i) ^ {r_i - r_{i-1} - 1}\right) \right]\\
        &\;\;\;\;\;= \sum_{r_0 < \cdots < r_{l}} \sum_{k_0 < \cdots < k_l} \Pr\left[A^\prime_{\vec{k}, \vec{r}}\right] (-1)^{a_{k_0}} \Psi(z_0)^{r_0 - 1} \left(\prod_{i=1}^{l} (-1) ^ {a_{k_i}} \Psi(z_i) ^ {r_i - r_{i-1} - 1}\right) \\
        &\;\;\;\;\;= \sum_{k_0 < \cdots < k_l} \sum_{r_0 < \cdots < r_{l}} \prod_{1 \leq i \leq l} \alpha_{k_i - k_{i-1}, r_i - r_{i-1}} (-1)^{a_{k_0}} \Psi(z_0)^{r_0 - 1} \left(\prod_{i=1}^{l} (-1) ^ {a_{k_i}} \Psi(z_i) ^ {r_i - r_{i-1} - 1}\right)\\
        &\;\;\;\;\;= \sum_{k_0 < \cdots < k_l} \prod_{1 \leq i \leq l} (-1) ^ {a_{k_i}} P_{\alpha_{k_i - k_{i-1}} - 1}\left(\Psi(z_i)\right) (-1)^{a_{k_0}} P_{\alpha^\prime_{k_0} - 1}\left(\Psi(z_0)\right)\\
        &\;\;\;\;\;= \prod_{0 \leq i \leq l} \overline{\phi}(\Psi(z)) P\left( \frac{1}{z_0} \right) \cdot \sum_{k_0 < k_1 < \cdots < k_l} (-1)^{a_{k_0}} z_0 ^ {k_0 - 1} \left(\prod_{i=1} ^ {l} (-1)^{a_{k_i}} z_i ^ {k_i - k_{i-1} - 1}\right)
    \end{aligned}
\end{equation}

\subsection{Proof of \texorpdfstring{\thref{thm:can_compute_g}}{Theorem \ref{thm:can_compute_g}} for Simple Characters}

\label{subsec:simple_characters}

In order to complete the proof of \thref{thm:can_compute_g} for the case when $f(x_1, \ldots, x_l) = \prod_i (-1) ^{x_i}$, we will use a property of the M\"{o}bius transformation $\Psi$ defined in \thref{lem:complex_connection}

\begin{lemma}
\thlabel{lem:mobius}
There are constants $c_1, c_2 \in (0, 1/20)$, depending only on $q, q^\prime$, such that for any sufficiently large $L$, if $\abs{\arg(z)} \leq c_1 /L$ and $\rho = 1 - 1/L^2$, then for $w = \Psi (\rho z)$,
\[
\abs{w} \leq 1 - \frac{c_2}{L^2}
\]

Furthermore, for $q < 1/2$, we have:
\[
\Psi(0) = \frac{q}{p} < 1
\]
\end{lemma}

\begin{remark}
For $n$ sufficiently large, \thref{lem:mobius} implies that for all $z \in \{(1 - n^{4/5} \log^6 n) e^{i\theta} \mid \abs{\theta} \leq n^{2/5}\}$, it holds that $\lvert \Psi(z) \rvert \leq 1 - c_2 n^{-4/5} \log ^ 6 n$ for some constant $c_2 > 0$.
\end{remark}
\begin{remark}
For $\varepsilon > 0 $ sufficiently small (but depending only on $q, q^\prime$), \thref{lem:mobius} implies that for any $z \in [1 - 2 \varepsilon, 1 - \varepsilon]$, it holds that $\lvert \Psi(z) \rvert \leq 1 - c_2 \varepsilon$ for some constant $c_2 > 0$.

Additionally, for $q < 1/2$, because $\Psi$ is continuous at $0$, for a sufficiently small $\varepsilon > 0$ depending only on $q, q^\prime$, for any $z\in\D$ such that $\abs{z} \leq \varepsilon$
\[
\abs{\Psi(z)} \leq 1 - \varepsilon.
\]
\end{remark}

\begin{proof}[Proof of \thref{lem:mobius}]
    Observe that $\phi$ is a M\"{o}bius transformation mapping $\mathbb{D}$ to a smaller disk which is contained in $\overline{\mathbb{D}}$ which is tangent to $\partial \mathbb{D}$ at $1$ and which maps $\mathbb{R}$ to $\mathbb{R}$.
    In particular, by linearising the map $\Psi$ at $z = 1$, that $\Psi(1 + \varepsilon) = 1 + a \varepsilon + O(\lvert \varepsilon\rvert^2)$ for $a > 1$ depending only on $q, q^\prime$.
    Writing $z = e^{i\theta}$, we have:
    
    \begin{equation}
        \begin{aligned}
            &w = \Psi(\rho e^{i\theta}) \\
            &\;\;\;\;\;\; = 1 + a (\rho e^{i\theta} - 1) + O\left(\abs{\rho e^{i\theta} - 1}^2\right)\\
            &\;\;\;\;\;\; = 1 + a \left((1 - L^{-2}) (1 + i\theta) - 1\right) + O\left(\theta^2 + L^{-4}\right)\\
            &\;\;\;\;\;\; = 1 - aL^{-2} + ia(1 - L^{-2})\theta  + O\left(\theta^2 + L^{-4}\right)
        \end{aligned}
    \end{equation}
    
    Therefore, we have have $\abs{w}^2 = (1 - aL^{-2})^2 \pm O(\theta^2 + L^{-4}) \leq (1 - c_2 L^{-2})^2$ for sufficiently small $c_1, c_2$.
    
    The last part of the claim is easy to verify.
\end{proof}

\begin{proof} [Proof of \thref{thm:can_compute_g} for Simple Characters]
We now consider a simplification of the formula from \thref{lem:complex_connection}:
    \begin{equation}
    \label{eq:approximating_g_simple}
        \begin{aligned}
            &\mathbb{E}\left[ \sum_{r_0 < \cdots < r_{l}} (-1)^{\widetilde{a}_{r_0}} \Psi(z_0)^{r_0 - 1} \left(\prod_{i=1}^{l} (-1) ^ {\widetilde{a}_{r_i}} \Psi(z_i) ^ {r_i - r_{i-1} - 1}\right) \right] \left(\prod_{0 \leq i \leq l} \overline{\phi}(\Psi(z_i))\right) ^ {-1} P\left( \frac{1}{z_0} \right) ^ {-1} \\
            &\;\;\;\;\;= \cdot \sum_{k_0 < k_1 < \cdots < k_l} (-1)^{a_{k_0}} z_0 ^ {k_0 - 1} \left(\prod_{i=1} ^ {l} (-1)^{a_{k_i}} z_i ^ {k_i - k_{i-1} - 1}\right)
        \end{aligned}
    \end{equation}
    
    Note that the right-hand-side of equation \eqref{eq:approximating_g_simple} is the value that we want to compute and the left-hand-side depends only on the traces.
    By bounding the coefficients of the left-hand-side, we can show that $\exp(h(n))$ samples suffice to approximate the right-hand-side to within $\pm \exp(-\Omega(h(n)))$.
    
    We begin with $P\left( \frac{1}{z_0} \right)$.
    By our assumption of the small shift inaccuracy $\eta = O(h(n))$, we have $\supp(S) \subseteq[a, a+\eta] = [a, a+O(h(n))]$ for some $a \leq n$.
    Therefore
    \[
    P(\zeta) = \zeta ^ a \wt{P}(\zeta)
    \]
    for some polynomial $\wt{P}$ of degree $\leq \eta$.
    
    In the setting of \thref{thm:can_compute_g}, we have $\abs{\frac{1}{z_0} - 1} = O(n^{-2/5})$, which implies
    \[
    \abs{\wt{P}\left(\frac{1}{z_0}\right) - 1} = \abs{\wt{P}\left(\frac{1}{z_0}\right) - \wt{P}(1)} \leq \abs{\left(\frac{1}{z_0}\right) ^ \eta - 1} = O(n^{-1/5})
    \]
    Inserting this into the triangle inequality, for sufficiently large $n$, we get:
    \[
    \wt{P}\left(\frac{1}{z_0}\right) \geq 1 - \abs{\wt{P}\left(\frac{1}{z_0}\right) - 1} = 1 -  O(n^{-1/5}) \geq \frac{1}{2}
    \]
    
    \vs
    Next, we note that from their definitions, it is clear that $\overline{\phi}(\psi(z)) = p > 0$ for $z = 1$ and that it is a continuous function near $1$.
    This implies that for sufficiently small $c_1 > 0$ and for any $z \in [1-c_1, 1]$, we have $\overline{\phi}(\psi(z)) > p/2 > 0$.
    Inserting this into the appropriate term in equation \eqref{eq:approximating_g_simple}, we have:
    \[
    \abs{\left(\prod_{0 \leq i \leq l} \overline{\phi}(\Psi(z_i))\right) ^ {-1}} \leq \left( \frac{p}{2}\right)^{l+1} = \exp(O(n^{1/5}))
    \]
    
    Finally, in \thref{lem:truncating_lhs} we will show that truncating the polynomial on the left-hand-side does has a negligible effect on its value.
    This allows us to truncate the left-hand-side of equation \eqref{eq:approximating_g_simple} to at most $n ^ {O(l)} = \exp(O(n^{1/5} \log n))$ terms each with a coefficient of absolute value $\leq 1$.
    
    Therefore, any false sample will shift the average by at most $\exp(-\Omega(h(n)))$ and the probability that we will have more than $n^{2}$ false samples is $\exp(-O(n^2)) = \exp(-\omega(n))$.
    Furthermore, in the absence of false samples, the distribution of the left-hand-side is bounded by $\exp(o(h(n)))$, so from a simple application of Chernoff's bound, averaging over $\exp(h(n))$ samples will suffice to give us its value to within $\exp(-\Omega(h(n)))$ with probability $1 - \exp(-\exp(\Omega(h(n)))) \geq 1 - \exp(-\omega(n))$.
\end{proof}

\begin{lemma}
\thlabel{lem:truncating_lhs}
Define the polynomial $r$ to be
\[
r(\zeta_0, \zeta_1, \ldots, \zeta_l) = \mathbb{E}\left[ \sum_{r_0 < \cdots < r_{l}} (-1)^{\widetilde{a}_{r_0}} \zeta_0^{r_0 - 1} \left(\prod_{i=1}^{l} (-1) ^ {\widetilde{a}_{r_i}} \zeta_i ^ {r_i - r_{i-1} - 1}\right) \right]
\]
and let $\wt{r}$ be its truncation to degree at most $n \log n$ on the first coordinate and degree at most $n$ on the rest
Define the polynomial $r$ to be
\[
\wt{r}(\zeta_0, \zeta_1, \ldots, \zeta_l) = \mathbb{E}\left[ \sum_{\substack{r_0 < \cdots < r_{l} \\ r_0 - 1 \leq n \log n \\ r_i - r_{i-1} - 1 \leq n}} (-1)^{\widetilde{a}_{r_0}} \zeta_0^{r_0 - 1} \left(\prod_{i=1}^{l} (-1) ^ {\widetilde{a}_{r_i}} \zeta_i ^ {r_i - r_{i-1} - 1}\right) \right]
\]

Then for sufficiently small $c_1 > c_2 > 0$ and for any $z_0 \in A$ and $z_1, \ldots, z_l \in [1 - c_1, 1 - c_2]$ (or $q < 1/2$ and $z_1, \ldots, z_l \in [-c_1, c_1]$), we have
\[
\abs{\wt{r}(\Psi(z_0), \ldots, \Psi(z_l)) -  r(\Psi(z_0), \ldots, \Psi(z_l))} \leq \exp(-\Omega(h(n)))
\]
\end{lemma}

\begin{proof} [Proof of \thref{lem:truncating_lhs}]
    From \thref{lem:mobius} and the remarks following, it follows that for some $c, c^\prime > 0$:
    \begin{equation}
        \begin{aligned}
            &\abs{\Psi(z_0)} \leq 1 - c n^{-4/5} \log^6 n
            &\abs{\Psi(z_i)} \leq 1 - c^\prime
        \end{aligned}
    \end{equation}
    
    \thref{lem:truncating_lhs} follows almost trivially.
    Consider the set of monomials of the form $m_{j, \bfd}(\bfz) = z_0 ^ j z_1 ^{d_1} \cdots z_l ^{d_l}$ with $d_1 + \ldots + d_l = d$.
    Each of these monomials has norm $\leq \exp(-\Omega(j n^{-4/5} \log^6 n + d))$ and there are at most $d^{O(l)}$ such monomials.
    Each monomial has a coefficient of norm $\leq 1$, so their total contribution is at most:
    \[
    \sum_{j > n \log n \vee d > n} \exp(-\Omega(j n^{-4/5} \log^6 n + d)) d^{O(l)} \leq \exp(-\Omega(n^{1/5} \log ^7 n))
    \]
\end{proof}

\subsection{Proof of \texorpdfstring{\thref{thm:can_compute_g}}{Theorem \ref{thm:can_compute_g}} for General Characters}
\label{subsec:general_characters}

Fix some $\omega \in \{\pm 1\}^l$ and let $z_0, z_1, \ldots, z_l$ be a point such that $z_1 = \cdots = z_l$.
Let $j_0 = 0$ and let $j_1 < \cdots < j_{l^\prime}$ be the indices for which $\omega_{j_i} = -1$ (for $\omega = (1, \ldots, 1)$, we set $l^\prime = 0$).

If $j_{l^\prime} < l$, then $f(x_1, \ldots, x_l) = \chi_\omega (x_1, \ldots, x_l)$ does not depend on the last $l - j_{l^\prime}$ coordinates.
In this case, setting $\wt{l} = j_{l^\prime}$, we can write:
\begin{equation}
\label{eq:g_l_prime}
    \begin{aligned}
        g_\bfx ^{\chi_\omega}(z_0, z_1, \ldots, z_l) &= \sum_{k_0 < \cdots < k_{l}}  (-1) ^{x_{k_0}} z_0 ^ {k_0 - 1} \omega_1 ^{x_{k_1}} z_1 ^{k_1 - k_0 - 1} \cdots \omega_{l} ^{x_{k_{l}}} z_{l} ^ {k_{l} - k_{l-1} - 1}\\
        &= \left(\sum_{k_0 < \cdots < k_{\wt{l}}}  (-1) ^{x_{k_0}} z_0 ^ {k_0 - 1} \omega_1 ^{x_{k_1}} z_1 ^{k_1 - k_0 - 1} \cdots \omega_{\wt{l}} ^{x_{k_{\wt{l}}}} z_{\wt{l}} ^ {k_{\wt{l}} - k_{\wt{l}-1} - 1}\right) \\
        & \;\;\;\;\;\;\;\;\;\;\;\;\;\;\;\;\;\;\;\;\;\;\;\;\;\;\;\;\;\;\;\; \cdot \left(\sum_{k_{\wt{l}+1} < \cdots < k_{l}}   z_{\wt{l}+1} ^{k_{\wt{l} + 1} - k_{\wt{l}} - 1} \cdots z_{l} ^ {k_{l} - k_{l-1} - 1}\right)\\
        &= \left(\sum_{k_0 < \cdots < k_{\wt{l}}}  (-1) ^{x_{k_0}} z_0 ^ {k_0 - 1} \omega_1 ^{x_{k_1}} z_1 ^{k_1 - k_0 - 1} \cdots \omega_{l} ^{x_{k_{l}}} z_{l} ^ {k_{\wt{l}} - k_{\wt{l}-1} - 1}\right) \cdot \left( \frac{1}{1 - z_1}\right)^{l - \wt{l}}
    \end{aligned}
\end{equation}

Because we assumed that $z_1 = \cdots = z_{l} \in [1 - c_1, 1 - c_2]$ (or $[-c_1, c_1]$), the second factor is $\exp(\Theta(l))$, and we focus on the first one.
For simplicity, we write the rest of our proof for $\wt{l} = l$, but it can be easily generalized to any $\wt{l} \leq l$.

For each $i$, the polynomial in equation \eqref{eq:g_l_prime} sums over all the possible sequences of $k_{j_{i-1}} < k_{j_{i-1}+1} \cdots < k_{j_{i}-1} < k_{j_{i}}$, despite the fact that they have the same coefficient depending only on $k_{j_{i-1}}, k_{j_{i}}$. Simplifying this summation, we have:
\begin{equation}
\label{eq:g_of_character}
    \begin{aligned}
        g_\bfx ^{\chi_\omega}(z_0, z_1, \ldots, z_l) &= \sum_{k_0 < \cdots < k_{l}} (-1) ^{x_{k_0}} z_0 ^ {k_0 - 1} \omega_1 ^{x_{k_1}} z_1 ^{k_1 - k_0 - 1} \cdots \omega_{l} ^{x_{k_{l}}} z_{l} ^ {k_{l} - k_{l-1} - 1}\\
        &= \sum_{k_0 < k_{j_1} \cdots < k_{j_{l^\prime}}} (-1) ^{x_{k_0}} z_0 ^ {k_0 - 1} \prod_{1 \leq i \leq l^\prime} (-1)^{x_{k_{j_i}}} z_{j_i} ^ {k_{j_i} - k_{j_{i-1}} - j_i + j_{i-1}} 
        \left(\begin{matrix}
            k_{j_i} - k_{j_{i-1}} - 1 \\ j_i - j_{i-1} - 1
        \end{matrix}\right)\\
    \end{aligned}
\end{equation}

Notice that if we were to multiply each member of the product in equation \eqref{eq:g_of_character} by $(j_i - j_{i-1} - 1)!$, we would see that this polynomial is a high-order derivative of a simpler function:

\begin{equation}
\label{eq:differentiation}
    \begin{aligned}
        g_\bfx ^{\chi_\omega}(\bfz) &= \sum_{k_0 < k_{j_1} \cdots < k_{j_{l^\prime}}} (-1) ^{x_{k_0}} z_0 ^ {k_0 - 1} \prod_{1 \leq i \leq l^\prime} (-1)^{x_{k_{j_i}}} z_{j_i} ^ {k_{j_i} - k_{j_{i-1}} - j_i + j_{i-1}} 
        \left(\begin{matrix}
            k_{j_i} - k_{j_{i-1}} - 1 \\ j_i - j_{i-1} - 1
        \end{matrix}\right)\\
        &= \frac{\partial^{l - l^\prime}}{\prod_{1 \leq i \leq l^\prime} (j_i - j_{i-1} - 1)! \partial \zeta_{j_i} ^ {j_i - j_{i-1} - 1}}\\
        & \;\;\;\;\;\;\;\;\;\;\;\;\;\;\;\;\;\;\;\;\;\;\;\;\;\;\;\;\;\;\;\;\sum_{k_0 < k_{j_1} \cdots < k_{j_{l^\prime}}} (-1) ^{x_{k_0}} \zeta_0 ^ {k_0 - 1} \prod_{1 \leq i \leq l^\prime} (-1)^{x_{k_{j_i}}} \zeta_{j_i} ^ {k_{j_i} - k_{j_{i-1}} - 1} \at[\big]{\zeta=\bfz}
    \end{aligned}
\end{equation}

We now note that the function being differentiated is equal to the $g$ for a simple character of length $l^\prime$:

\begin{equation}
    \begin{aligned}
        &\sum_{k_0 < k_{j_1} \cdots < k_{j_{l^\prime}}} (-1) ^{x_{k_0}} \zeta_0 ^ {k_0 - 1} \prod_{1 \leq i \leq l^\prime} (-1)^{x_{k_{j_i}}} \zeta_{j_i} ^ {k_{j_i} - k_{j_{i-1}} - 1} \at[\big]{\zeta=\bfz}\\
        & \;\;\;\;\;\;\;\;\;\;\;\;\;\;\;\;\;\;\;\;\;\;\;\;\;\;\;\;\;\;\;\;=\sum_{k_0 < k_1 \cdots < k_{{l^\prime}}} (-1) ^{x_{k_0}} \zeta_0 ^ {k_0 - 1} \prod_{1 \leq i \leq l^\prime} (-1)^{x_{k_{i}}} \zeta_{j_i} ^ {k_{i} - k_{{i-1}} - 1} \\
        & \;\;\;\;\;\;\;\;\;\;\;\;\;\;\;\;\;\;\;\;\;\;\;\;\;\;\;\;\;\;\;\;= g_{\bfx} ^ {\prod_i (-1)^{x_i}} (\zeta_0, \zeta_{j_1}, \ldots, \zeta_{j_{l^\prime}})
    \end{aligned}
\end{equation}

In Section \ref{subsec:simple_characters}, we showed that we can compute
\[
g_{\bfx} ^ {\prod_i (-1)^{x_i}} (\zeta_0, \zeta_{j_1}, \ldots, \zeta_{j_{l^\prime}})
\]
to a high degree of accuracy in a neighborhood of $z_0, z_1, \ldots, z_{l^\prime}$.
In order to compute $g_\bfx ^{\chi_\omega}(\bfz)$, we show that we can use a simple interpolation technique to perform the differentiation shown in equation \eqref{eq:differentiation} numerically.
In particular, we prove the two following lemmas from which the proof of this claim follows immediately.

\begin{lemma}
\thlabel{lem:interpolation}
Let $c, \delta > 0$ be some real parameters and let $P$ be an oracle that computes for a given point $z_1, \ldots, z_l \in [-c, c] ^ l$ the value of some polynomial $p$ of degree at most $n$ at the given point, up to some additive error $\delta > 0$.
Let $\bfj = (j_1, \ldots, j_l)$ be some vector of integers (all smaller than $n$), define $m_{\bfj} = z_1 ^{j_1} \cdots z_l ^{j_l}$ be the $\bfj$th monomial and $j_\textup{tot} = \sum_i j_i$.

Given $\poly(n, 1/c)^{O(l+j_\textup{tot})}$ queries to $P$, we can compute the coefficient of $m_\bfj$ to within an additive error of $\poly(n, 1/c)^{O(l+j_\textup{tot})}\delta$ in time $\poly(n, 1/c)^{O(l+j_\textup{tot})}$.
\end{lemma}

The proof of \thref{lem:interpolation} is shown in Appendix \ref{app:interpolation}.

\begin{lemma}
\thlabel{lem:rhs_truncation}
Define the polynomial $r$ to be
\[
r(z_0, \ldots, z_l) =  \sum_{k_0 < k_1 < \cdots < k_l} (-1)^{a_{k_0}} z_0 ^ {k_0 - 1} \left(\prod_{i=1} ^ {l} (-1)^{a_{k_i}} z_i ^ {k_i - k_{i-1} - 1}\right)
\]
and let $\wt{r}$ be its truncation to degree $\leq n$ on the coordinates $z_1, \ldots, z_l$:
\[
\wt{r}(z_0, \ldots, z_l) = \sum_{\substack{k_0 < k_1 < \cdots < k_l \\ k_{i} - k_{i-1} - 1 \leq n}} (-1)^{a_{k_0}} z_0 ^ {k_0 - 1} \left(\prod_{i=1} ^ {l} (-1)^{a_{k_i}} z_i ^ {k_i - k_{i-1} - 1}\right)
\]

Then for $z_0, z_1, \ldots, z_l$ as defined above, we have
\[
\abs{\wt{r}(z_0, \ldots, z_l) - r(z_0, \ldots, z_l)} \leq \exp(-\Omega(n))
\]
\end{lemma}

Combining \thref{lem:interpolation}, which states that the $\bfj$th derivative of any degree $\leq n$ can be approximated without a significant increase to the inaccuracy, with \thref{lem:rhs_truncation}, which states that $g_\bfx ^ {\prod_i (-1) ^ {x_i}}$ can be approximated to a very high degree of accuracy by a degree $\leq n$ polynomial and the fact that our target polynomial $g_\bfx ^{\chi_\omega}$ is a derivative of $g_\bfx ^ {\prod_i (-1) ^ {x_i}}$ completes our proof.

\begin{proof} [Proof of \thref{lem:rhs_truncation}]
The proof of \thref{lem:rhs_truncation} is very similar to our proof of \thref{lem:truncating_lhs}.

Fix $z_0$ as in \thref{lem:rhs_truncation}, we now view $r$ and $\wt{r}$ as polynomials in $z_1, \ldots, z_l$ whose coefficients were set as functions of $z_0$.
The absolute value of each of those coefficients is trivially bounded from above by $\sum_i \abs{z_0}^i < n^{4/5}$.

Consider the total contribution of the monomials of the form $z_1^{d_1} \cdots z_l {d_l}$ of total degree $\sum_i d_i = d$.
There are $d ^ {O(l)}$ such monomials, each has coefficient $\leq n^{4/5}$ and size $\abs{z_i}^d = \exp(-\Omega(d))$.
Therefore the total contribution of coefficients of total degree $d \geq n$ is at most
\[
\sum_{d \geq n} d^{O(l)} n^{4/5} \exp(-\Omega(d)) \leq \sum_{d \geq n} \exp(-\Omega(d)) = \exp(-\Omega(n)).
\]

\end{proof}

\subsection{Proof of \texorpdfstring{\thref{thm:can_compute_g}}{Theorem \ref{thm:can_compute_g}} for General Functions}

\begin{proof}
    We begin by rewriting $f(x_1, \ldots, x_l)$ as its Fourier transformation over $\mathbb{F}_2 ^l$:
\[
f(\bfx) = \sum_{\omega \in \{1, -1\}^l} \hat{f}(\omega) \chi_\omega (\bfx)
\]
where $\chi_\omega(\bfx) = \prod_{1 \leq i \leq l} \omega_i^{x_i}$.

We now use the additivity of $g^f_\bfx$ (as a function of $f$) to write:
\[
g_\bfx ^ f (z_0, z_1, \ldots, z_l) = \sum_{\omega \in \{1, -1\}^l} \hat{f}(\omega) g_\bfx ^{\chi_\omega}(z_0, z_1, \ldots, z_l)
\]

In Section \ref{subsec:general_characters}, we showed that we can approximate $g_\bfx ^{\chi_\omega}(z_0, z_1, \ldots, z_l)$ to within $\pm \exp(-\Omega(h(n)))$ for any character $\omega \in \{-1, 1\}^l$.
Parceval's theorem easily bounds the norms of the $\hat{f}(\omega)$ coefficients to at most $\exp(O(l)) = \exp(o(h(n)))$.
Therefore, we can combine the results of these approximations to obtain a high accuracy approximation of $g_\bfx ^ f (z_0, z_1, \ldots, z_l)$ from the traces.

\end{proof}

\section{Proof of \texorpdfstring{\thref{thm:worst_case}}{Theorem \ref{thm:worst_case}}}

\label{sec:analysis}
In Section \ref{sec:fourier}, we showed that for any function $f$ from $\{0, 1\}^l$ to the unit disk, we can map it into a polynomial related to the input message which can be approximated to a high degree of accuracy from the traces.
In this section, we will construct a function $f$ for which our approximation of $g_\bfx ^f$ shown in \thref{thm:can_compute_g} will suffice to reconstruct the $n+1$th bit of $\bfx$, proving \thref{thm:worst_case}.

For his upper bound, Chase proved used a lemma that members of the class of polynomials defined below reaches non-negligible values on a small subarc.
We will prove that a similar polynomial has non-negligible values on a small sub-arc of radius $1 - \varepsilon$.

\begin{theorem} [Adaptation of Theorem 5 of \cite{chase2021separating}]
\thlabel{thm:like_borwein}
    Let $\mathcal{P}_n^\mu$ denote the set of polynomials of the form $p(x) = \zeta - \eta x^d + \sum_{n^\mu \leq j \leq n} a_j x^j$ where $\eta \in \{0, 1\}$, $\zeta \in \partial \D$ and $\abs{a_j} \leq 1$.
    
    For any $\mu \in (0, 1)$, there exists some constant $C_1 > 0$, such that for all sufficiently large $n$, any $p \in \mathcal{P}_n ^ \mu$, it holds that for every $\rho \in [0, 1]$:
    \[
    \max_{\lvert \theta \rvert \leq n^{-2\mu}} \lvert p(\rho e^{i\theta}) \rvert \geq \exp\left(-C_1 n^\mu \log^5 n\right)
    \]
\end{theorem}

\thref{thm:like_borwein} is a generalization of Theorem 5 in Chase's upper bound (Chase proves this for $\zeta = \rho = 1$).
Our proof of the general theorem is similar to Chase's proof and we show it in Appendix \ref{app:proof_of_chase}.
Throughout the rest of this section, we will prove that \thref{thm:worst_case} follows from \thref{thm:like_borwein}.

\subsection{Corollaries of \texorpdfstring{\thref{thm:like_borwein}}{Theorem \ref{thm:like_borwein}}}
\label{subsec:corollaries}

We will use \thref{thm:like_borwein} for $\mu = 1/5$ and with $\rho = 1 - n^{-4\mu} \log^6 n$.
For the rest of this section, we set $\rho = 1 - n^{-4/5} \log^6 n$, $l = 2n^{1/5} + 1$, and $A = \{\rho e^{i\theta} \mid \lvert \theta \rvert \leq n^{2/5} \}$.
A corollary of \thref{thm:like_borwein} is the following:

\begin{corollary} [Adaptation of Proposition 6.3 from \cite{chase2021separating}]
\label{cor:good_z0}
    For some constant $C > 0$, let $\bfx, \bfy \in \{0, 1\}^\N$ be binary strings, such that $\bfx, \bfy$ agree on their first $n$ bits but not on their $(n+1)$th bit.
    Then there exist some $\bfw \in \{0, 1\}^{l}$ and $z_0 \in A$ such that
    \[
    \abs{\sum_{k} \left[ (-1)^{x_k} 1_{\bfx(k+1:k+l) = \bfw} - (-1){y_k} 1_{\bfy(k+1:k+l) = \bfw} \right] z_0 ^k} \geq \exp \left(- n^{1/5} \log ^ 6 n\right) \exp(-C n^{1/5} \log ^5 n)
    \]
\end{corollary}

\begin{proof} [Proof of Corollary \ref{cor:good_z0}]
    Set $\bfw^\prime = \bfx(n - l:n-1)$.
    
    Like Chase, we note that Lemmas 1 and 2 of \cite{robson1989separating} and the fact that either either $\bfw^\prime0$ or $\bfw^\prime1$ has no period of length $\leq n^{1/5}$ imply that for the right choice of $\bfw \in \{\bfw^\prime0, \bfw^\prime1\}$, the indices $k$ for which $\bfx(k:k+l) = \bfw$ are $n^{1/5}$ separated.
    
    Define $q(z)$ to be the polynomial
    \[
        q(z) = \sum_k \left[(-1)^{x_k} 1_{\bfx(k+1:k+l) = \bfw} - (-1)^{y_k} 1_{\bfy(k+1:k+l) = \bfw}\right] z^{k-1}
    \]
    In Chase's construction, the polynomial $p(z) = q(z) / z^{n-l} \in \mathcal{P}_n^{1/5}$, is used an input to \thref{thm:like_borwein}
    
    However, in our case, this sum is infinite, so we cannot directly apply \thref{thm:like_borwein} to it.
    In order to do so, we will consider its truncation
    \[
        \widetilde{p}(z) = \sum_{k \leq 2n - l} \left[(-1)^{x_k} 1_{\bfx(k+1:k+l) = \bfw} - (-1)^{y_k} 1_{\bfy(k+1:k+l) = \bfw}\right] z^{k-m}
    \]
    (where $m = n - l$, and clearly $\widetilde{p} \in \mathcal{P}_{n}^{1/5}$)
    
    We will be evaluating $\widetilde{p}(\cdot)$ at points $\abs{z_0} = \rho$, so it is easy to show from the triangle inequality that:
    \[
    \abs{p(z_0)} \geq \abs{\widetilde{p}(z_0)} - \sum_{k \geq n} \abs{z_0^k} \geq \abs{\widetilde{p}(z_0)} - \exp\left(-n^{1/5} \log^6 n\right)
    \]
    
    Therefore, by \thref{thm:like_borwein}, there exists some $\theta \in [-n^{-2/5}, n^{-2/5}]$, such that for $z_0 = \rho e^{i \theta}$:
    \begin{equation}
        \begin{aligned}
            &\abs{q(z_0)} = \abs{z_0 ^ {n - l}} \abs{p(z_0)} \geq \rho^{n-l} \abs{\exp(-C_1 n^{1/5} \log^{5} n) - \exp(-n^{1/5} \log^6 n)}\\
            &\;\;\;\;\;\;\;\;\geq \exp(- n^{1/5} \log^{6} n) \exp(-C n^{1/5} \log^{5} n)
        \end{aligned}
    \end{equation}
\end{proof}

For any string $\bfx$ (and where $\bfw$ is implied from context), define the polynomial $g_\bfx(\cdot)$ to be:
\[
g_{\bfx}(\zeta_0 , \zeta_1, \ldots , \zeta_l) = \sum_{k_0 < \cdots < k_{l-1}} (-1)^{x_{k_0}} \prod_{1 \leq i \leq l} 1_{x_{k_i} = w_i} \zeta_0 ^{k_0 - 1} \zeta_1 ^{k_1 - k_0 - 1} \ldots \zeta_{l-1}^{k_{l}-k_{l-1} - 1}
\]

For any choice of $\bfw$, note that the left-hand-side of the claim in Corollary \ref{cor:good_z0} is equal to 
\[
\abs{g_\bfx(z_0, 0, \ldots, 0) - g_\bfy(z_0, 0, \ldots, 0)}
\]

This implies that setting $z_1, \ldots, z_l$ to $0$ would give us a point $z_0$ on the arc $A$, where these polynomials differ.
Indeed when the deletion probability $q$ is below $1/2$, we will separate between these polynomials by using the traces to estimate the evaluation of $g_\bfx$ at the point $(z_0, 0, \ldots, 0)$ where $z_0$ and $\bfw$ are as promised in Corollary \ref{cor:good_z0}.
Evaluating at points like this will allow us to reconstruct $\bfx$ more efficiently by using the sparsity of the $1$-variable polynomial $f(z_0) = g(z_0, 0, \ldots, 0)$.

However, as in Chase's construction, this point is difficult to estimate when $q \geq 1/2$, so in order to reconstruct $\bfx$ from such channels, we will estimate the evaluation of $g_\bfx(z_0, z_1, \ldots, z_l)$ at a point where point $z_1 = \cdots = z_l \in [1 - \varepsilon_1, 1 - \varepsilon_2]$ (for small positive $\varepsilon_{1,2}>0$).

\begin{corollary} [Adaptation of Corollary 6.1 from \cite{chase2021separating}]
\label{cor:good_zs}
    Let $\varepsilon_1 > \varepsilon_2 > 0$ be some positive constants.
    For some $C^\prime > 0$, let $\bfx, \bfy \in \{0, 1\}^\N$ be binary strings, such that $\bfx, \bfy$ agree on their first $n - 1$ bits but not on their $n$th bit.
    Then there exist some $\bfw \in \{0, 1\}^{l}$, $z_0 \in A$ and $z_1 = \cdots = z_l \in [1 - \varepsilon_1, 1-\varepsilon_2]$, such that
    \begin{equation*}
        \begin{aligned}
            &\abs{\sum_{k_0 < \cdots < k_l} \left((-1)^{x_{k_0}}\prod_{1 \leq i \leq l}\left[ 1_{x_{k_i} = w_i}\right] - (-1)^{y_{k_0}}\prod_{1 \leq i \leq l}\left[ 1_{y_{k_i} = w_i}\right] \right) z_0 ^{k_0 - 1} z_1 ^{k_1 - k_0 - 1} \ldots z_{l-1} ^ {k_{l} - k_{l-1} - 1} } \\
            &\;\;\;\;\;\;\;\;\;\;\;\;\;\;\;\;\;\;\;\;\;\;\;\;\;\;\;\;\;\;\;\;\;\;\;\;\;\;\;\;\;\;\;\;\;\;\;\;\;\;\;\;\;\;\;\;\;\;\;\;\;\;\;\;\;\;\;\;\;\;\;\;\;\;\;\;\;\;\;\;\;\;\;\;\;\;\;\;\;\;\;\;\;\;\;\;\geq \exp \left(-C^\prime n^{1/5} \log ^ 6 n\right)
        \end{aligned}
    \end{equation*}
\end{corollary}

\begin{proof} [Proof of Corollary \ref{cor:good_zs}]
    The proof of this corollary will follow the same lines as the proof of Corollary 6.1 in \cite{chase2021separating}.
    
    Let $\bfw$ and $z_0$ be those promised in Corollary \ref{cor:good_z0}.
    We define:
    \[
    f(z_1) = (1-\rho) \left(\begin{matrix}
        n\\l
    \end{matrix}\right)^{-1} \sum_{k_0 < \cdots < k_{l} } \left((-1)^{x_{k_0}}\prod_{1 \leq i \leq l}\left[ 1_{x_{k_i} = w_i}\right] - (-1)^{y_{k_0}}\prod_{1 \leq i \leq l}\left[ 1_{y_{k_i} = w_i}\right] \right) z_0 ^{k_0 - 1} \left((1 - \varepsilon_1) z_1\right) ^{k_{l} - k_0 - l}
    \]
    
    As in Chase's proof, $f$ is a polynomial in $z_1$ we will show that each of its coefficients is upper bounded by $1$ in absolute value.
    
    We first note that $\sum_{k_0} \abs{z_0^{k_0}} = \frac{1}{1-\rho}$, so the absolute value of the contribution of the summation over $k_0$ is bounded.
    
    Denote the power of $z_1$ by $n^\prime \geq k_{l} - k_0 - l$.
    When $n^\prime \leq n$ it is trivial that the absolute value of the coefficient of $z_1 ^ {n^\prime}$ is at most $1$, since the normalization factor is at most $1$ over the number of sets $k_1, \ldots, k_l$ maintaining this equality.
    When $n^\prime > n$, the number of values of $k_1, \ldots, k_l$ is $\exp(O(l \log n^\prime))$ and the $(1 - \varepsilon_1)^{n^\prime} = \exp(-\Omega(n^\prime))$, so the coefficient is still bounded.
    
    Furthermore, this shows that when evaluating $f(z_1)$ at points where $\abs{z_1} \leq 1 - \varepsilon_2$, the contribution of the terms with power $n^\prime \geq n$ is exponentially small.
    Define $g$ to be the truncated version of $f$:
    \[
     g(z_1) = (1-\rho) \left(\begin{matrix}
        n\\l
    \end{matrix}\right)^{-1} \sum_{\substack{k_0 < \cdots < k_{l-1} \\ k_{l-1} - k_0 - l+1 \leq n} } \left[1_{x_{k_i} = w_i} - 1_{y_{k_i} = w_i}\right] z_0 ^{k_0 - 1} \left((1 - \varepsilon_1) z_1\right) ^{k_{l-1} - k_0 - l+1}
    \]
    
    Clearly:
    \begin{equation}
    \label{eq:clearly_g_and_f}
        \abs{g(z_1) - f(z_1)} \leq \exp(-\Omega(n))
    \end{equation}
    Therefore, applying Theorem 5.1 of \cite{borwein1999littlewood}, we have the following inequality:
    \begin{equation}
    \label{eq:g_is_large}
        \begin{aligned}
            &\left(\begin{matrix} n\\l \end{matrix}\right) \frac{1}{1-\rho} \max_{z_1 \in[1-\varepsilon, 1]} \abs{g(z_1)} \geq \left(\begin{matrix} n\\l \end{matrix}\right) \frac{1}{1-\rho} \abs{g(0)}^{c_1 / \varepsilon} \exp(- \frac{c_2}{\varepsilon}) \\
            &\;\;\;\;\;\;\;\; \geq \left(\begin{matrix} n\\l \end{matrix}\right) \frac{1}{1-\rho} \left( (1 - \rho) \left(\begin{matrix} n\\l \end{matrix}\right)^{-1} \exp(-C n^{1/5} \log^{6} n) \right)^{c_1 / \varepsilon} \exp(- \frac{c_2}{\varepsilon})\\
            &\;\;\;\;\;\;\;\; \geq \exp(-C^\prime/2 n^{1/5} \log^{6} n)
        \end{aligned}
    \end{equation}
    where $C^\prime > 0$ is some constant and
    \[
    \varepsilon = 1 - \frac{1-\varepsilon_2}{1 - \varepsilon_1} > 0
    \]
    
    Combining equations \eqref{eq:clearly_g_and_f} and \eqref{eq:g_is_large} yields our claim.
\end{proof}

\subsection{Completing the Proof}

We are now ready to complete our proof of \thref{thm:worst_case}.
Assuming that we know the first $n$ bits of $\bfx$, we will show that one can reconstruct its $n+1$th bit from the traces.

Let $C>0$ be a sufficiently large constant.
We will do this by enumerating over all the pairs of options $\textbf{o} \in \{0, 1\}^{Cn - n - 1}$ of $\bfx(n+1:Cn)$.
Let $\textbf{o}^0, \textbf{o}^1$ be such a pair, and consider the strings $\bfy^0 = \bfx(0:n-1)0\textbf{o}^0,\bfy^1 = \bfx(0:n-1)1\textbf{o}^1$.

We will then combine \thref{thm:can_compute_g}, which shows that we can approximate $g_\bfx ^ f (z_0, z_1, \ldots, z_l)$, and \thref{thm:like_borwein} with Corollaries \ref{cor:good_z0} and \ref{cor:good_zs}, which show that $g_\bfx ^ f (z_0, z_1, \ldots,  z_l)$ strongly depends on the $n+1$th bit of $\bfx$.
We use this combination to define a Boolean test $\tau$ such that if $\bfy^b = \bfx(:Cn)$ (for $b \in \{0, 1\}$), then $\tau$ will return $b$ when run on $\bfy^0, \bfy^1$ (with very high probability).

Repeating this test for each such pair $\bfy^0, \bfy^1$, for the correct assignment $\bfy^b = \bfx(:Cn)$, the test $\tau$ will always return $b = x_n$.
There can be no assignment $\bfy^{1 - x_n}$ for which $\tau$ always returns $1 - x_n$, because when matched with $\bfy^{x_n} = \bfx(:Cn)$, the test will return $x_n$ by this very property.

Finally, in Section \ref{subsec:efficient_enumeration} we will show that when $q < 1/2$, this enumeration can be carried out more efficiently (in time $\exp(o(n))$).

\vs

Let $\bfy^0, \bfy^1$ be some pair of strings as above.
If $q < 1/2$, let $z_0$ and $\bfw$ be those promised by Corollary \ref{cor:good_z0} for separating between $\bfy^0$ and $\bfy^1$ and let $z_1, \ldots, z_l = 0$.
If $q \geq 1/2$, let $z_0, z_1, \ldots, z_l$ and $\bfw$ be those promised by Corollary \ref{cor:good_zs} for separating between $\bfy^0$ and $\bfy^1$.

Define $f(x_1, \ldots, x_l) \defeq \prod_i 1_{x_i = w_i}$, and define the polynomial:
\[
g_{\bfx}(z_0 , z_1, \ldots , z_l) = \sum_{k_0 < \cdots < k_l} (-1)^{x_{k_0}} f(x_{k_1}, \ldots , x_{k_l}) z_0 ^{k_0 - 1} z_1 ^{k_1 - k_0 - 1} \ldots z_{l}^{k_l-k_{l-1} - 1}
\]

The test $\tau$ is defined as follows.

First $\bfw$ and $z_0, z_1, \ldots, z_l$ are selected as in Corollary \ref{cor:good_zs} (or Corollary \ref{cor:good_z0} when $q < 1/2$), when applied to the strings $\bfy^0, \bfy^1$.
\thref{thm:like_borwein} is not constructive, but our choice of $\bfw$ was given directly from the first $n$ bits of $\bfy^0, \bfy^1$ and given the polynomial $g_{\bfy^0} - g_{\bfy^1}$, we can take a grid of values of $z_0, z_1, \ldots, z_l$ within the allowed region with a sufficiently small ($\exp(-\wt{O}(l))$) distance between them and choose a point for which $\abs{g_{\bfy^0} - g_{\bfy^1}}$ is sufficiently large ($\exp(-O(n^{1/5} \log^5 n))$).

We know that such a point exists from Corollary \ref{cor:good_zs} (or \ref{cor:good_z0}), and in Section \ref{subsubsec:truncation} we will show that truncating $g_\bfx$ to the $Cn$th power changes it by at most $\exp(-\omega(n^{1/5} \log^5 n))$.
Between the points of a grid with jumps $\exp(-\omega(n^{1/5} \log^5 n))$, the polynomial $g_{\bfy^0} - g_{\bfy^1}$ cannot change too much which implies that for the grid points closest to the point promised by Corollary \ref{cor:good_zs} (or \ref{cor:good_z0}), the difference between $g_{\bfy^0}$ and $g_{\bfy^1}$ is also $\exp(-O(n^{1/5} \log^5 n))$.

Given this point, we use the algorithm promised in \thref{thm:can_compute_g} to approximate $g_{\bfx}(z_0, z_1, \ldots, z_l)$.
The test $\tau$ outputs the value $b \in \{0, 1\}$ for which $g_{\bfy^b} (z_0, z_1, \ldots, z_l)$ is closest to our approximation of $g_\bfx$ at that point.
If $\bfy^b = \bfx(:Cn)$ for some $b \in \{0, 1\}$, then we will clearly output $b$, since our approximation of $g_{\bfx}(z_0, z_1, \ldots, z_l)$ is at most $\exp(-\Omega(n^{1/5} \log^7 n))$ from the correct value which is equal to $g_{\bfy^b}$ (and far from $g_{\bfy^{1-b}}$).

\subsubsection{Truncating \texorpdfstring{$g_{\bfx}$}{g bfx} Causes only a Small Change}
\label{subsubsec:truncation}

Our estimation method assumes we are somehow able to efficiently perform the test described above for any pair $\bfy^0, \bfy^1$ and select the one for which the test passes.
However, there is an infinite number of possibilities for $\bfx$, as we assumed it was a string of infinite length.

In this section, we will show that only a short prefix of $\bfx$ can have a non-negligible effect on the test.
This implies that it is enough to use the value of $\bfx$ in only a finite number of indices in order to apply this bit-recovery algorithm in a finite time.

The bound we give here will bound the time-complexity of the bit recovery by $\exp(O(n))$, since we will only bound the number of bits of $\bfx$ over which we enumerate by $Cn$, which the average to worst case reduction would translate to a bit recovery with a polynomial complexity of $\exp(O(\log n))$.
In Section \ref{subsec:efficient_enumeration} we will show that when $q \leq 1/2$, it is enough to enumerate only over a small part of the entropy of these first $Cn$ bits, resulting in a $\exp(\widetilde{O}(n^{4/5}))$ complexity.

\begin{lemma}
\thlabel{lem:only_first_coordinates}
Define $\widetilde{g}_\bfx$ to be the truncation of $g_\bfx$ that has only monomials with total degree $\leq Cn - l$:
\[
\widetilde{g}_\bfx(z_0 , \zeta_1, \ldots , z_{l}) = \sum_{k_0 < \cdots < k_{l} < n + n \log n} (-1)^{x_{k_0}} f(x_{k_1}, \ldots , x_{k_{l}}) z_0 ^{k_0 - 1} z_1 ^{k_1 - k_0 - 1} \cdots z_{l}^{k_{l}-k_{l-1} - 1}
\]

Let $C ^ \prime$ be the constant from Corollary \ref{cor:good_zs}. 
Then, for sufficiently large $C>0$
\[
\abs{g_\bfx(z_0, \ldots,  z_l) - \widetilde{g}_\bfx(z_0, \ldots, z_l)} < \frac{1}{10} \exp(-C^\prime n^{1/5} \log^6 n)
\]
\end{lemma}

\begin{proof}
    This lemma follows trivially from the fact that the norms of the entries $z_i$ are bounded sufficiently below $1$.
Indeed, for $i > 1$, all of the $z_i$ are of norm $\leq 1 - \varepsilon$ for some constant $\varepsilon > 0$, and $z_0$ is of norm $\leq 1 - n^{-4/5} \log^6 n$
For any $d_0, m$, limiting ourselves to monomials where $z_0$ has degree $d_0$ and the contribution of $z_1, \ldots, z_{l}$ to the degree is $m$ (that is, we are looking at the monomials of the form $z_0 ^ {d_0} z_1 ^{d_1} \cdots z_l ^{d_l}$ with $\sum_{i\geq 1} d_i = m$), there are $O((m+l)^l)$ such monomials, but the norm of any such monomial is at most $(1 - \varepsilon)^m = \exp(-\Omega(m))$ and its coefficient in $g_\bfx$ is of norm at most $1$.
Therefore, their total contribution decays exponentially with $m$ and is easy to see from a summations over this exponential decay, that truncating the ones with $m > Cn / 2$ changes $g_\bfx (z_0, \ldots, z_l)$ by at most $\exp(-\Omega(Cn))$.

Similarly for $i = 0$ we have $\abs{z_0} = 1 - n^{-4/5} \log^6 n$.
Therefore, the total contribution of powers $d_0 \geq Cn/2$ can also be bounded with the triangle inequality and the sum of a geometric series, and we see that their contribution to $g_\bfx (z_0, \ldots , z_{l})$ is at most $\exp(-\Omega(Cn^{1/5} \log^6 n))$.
Selecting a sufficiently large $C$, we prove prove the claim.
\end{proof}

This gives an algorithm for the shifted trace reconstruction with a complexity of $\leq \exp(C n)$ for some constant $C$.

\subsubsection{Efficient Enumeration when \texorpdfstring{$q < 1/2$}{q < 1/2}}
\label{subsec:efficient_enumeration}

The last portion of our claim that we need to prove is that the enumeration can be performed in time $\exp(O(n^{4/5} \log n))$, when $q < 1/2$.
To do this, we first note that our separation between $\bfy^0$ and $\bfy^1$ was based entirely on the values of the polynomial $P_{\bfy^0} (z_0) = g_{\bfy^0} (z_0, 0, \ldots, 0)$ and $P_{\bfy^1} (z_0) = g_{\bfy^1} (z_0, 0, \ldots 0)$.

This polynomial is sparse and can be written as:
\[
P_{\bfy} (z) = \sum_{k \geq 1} z^{k-1} 1_{\bfy(k:k+l) = \bfw}
\]

Next, we note that the polynomial $P_{\bfy}$ is determined completely by the indices in the string $\bfy$ where $\bfw$ appears as a consecutive substring.
By our very design of $\bfw$, these indices are $n^{1/5}$-separated.

Therefore, there are only 
\[
\left(\begin{matrix} Cn \\ Cn^{4/5} \end{matrix}\right) = \exp (O(n^{4/5} \log n))
\]
options for the truncation of $P_{\bfy}$ to its first $Cn$ powers and as we showed in Section \ref{subsubsec:truncation}, it is enough to determine $g_{\bfy}$ only up to its first $Cn$ powers.

Finally, we note that our choice of $\bfw$ depended only on the first $n - 1$ bits of $\bfx$, so we did not need to enumerate over any bits to compute it, and our choice of $z_0, z_1, \ldots, z_l$ depended only on the polynomial $P_{\bfy}$ also requiring no additional enumeration.

Therefore it will suffice to enumerate over only $\exp(O(n^{4/5}\log n)) = \exp(o(n))$ options, proving the complexity bound for $q < 1/2$.


\section*{Acknowledgements}
I would like to thank Aviad Rubinstein and Roni Con for their helpful comments on previous versions of this paper.
I would also like to thank Nina Holden, Robin Pemantle, Yuval Peres and Alex Zhai for their help in understanding their paper.

\bibliographystyle{alpha}
\bibliography{main}


\appendix

\section{Proof of Lemmas \ref{lem:coarsely_well_behaved} and \ref{lem:finely_well_behaved}}
\label{app:alignment}

\begin{proof} [Proof of \thref{lem:coarsely_well_behaved}]
    For any interval $I \subseteq [0, n]$, by \thref{lem:most_robust}, we know that $\bfx(I)$ is robust at scale $\lambda_c$ with probability $1 - e^{-\Omega(\ell / \lambda)} = 1 - e^{-\Omega(C^{1/2} \log n)} \geq 1 - n^{-4}$ for sufficiently large $C$.

    Similarly, from \thref{lem:within_interval} we have that $\Pr[\mcQ_{\ell_c, \lambda_c}(I, [0, n])] \leq e^{-\Omega(C^{1/2} \log n)} \leq n^{-6}$, where the probability measure is over both random strings $\bfx$ and the randomness of the channel $\omega$.

    Therefore, using Markov's inequality, we have that 
    \[
    \Pr_{\bfx} [\Pr_{\omega} [\mcQ_{\ell, \lambda} (I, [0, n]) \geq n^{-2}]] \leq n^{-4}
    \]

    Taking a union bound on all $O(n)$ choices of $I$ completes the proof.
\end{proof}

\begin{proof} [Proof of \thref{lem:finely_well_behaved}]
    Throughout the proof of this claim, we allow the implicit constants in $\Omega(\cdot)$ and $O(\cdot)$ to depend on $c_0$, but not on $C$.
    Fix a particular interval $J = [a, a+C\log n] \subseteq [0, n]$ of length $C \log n$, and let $\ell = \ell_f, \lambda = \lambda_f$ be as defined above.

    Consider the $m = \frac{1}{3} C^{1/3} \log n / h(C \log n)$ disjoint length $\ell$ intervals:

    \[
    I_1, \ldots, I_m \subset [a + \frac{1}{3} C \log n, a + \frac{2}{3} C \log n]
    \]

    For a given realization of $\bfx$, we say that $I_i$ is {\em bad} if either it does not have a clear robust bias at scale $\lambda$ or it holds that 
    \begin{equation}
    \label{eq:spurious_match}
        \Pr_\omega [\mcQ_{\ell, \lambda} (I_i, J)] \geq \exp (-c_0 C^{7/12} h(C \log n))
    \end{equation}

    Let $I_i$ be a bad segment and let $\ell ^ \prime = C^{1/6} \ell$, and define the event:
    \[
    H = \left\{
        \begin{matrix}
            \text{there exists some }t\text{ such that}\\
            \;\;\;\;g(t), g(t+\ell) \in I_i\text{ and } \lvert g(t) - g(t+\ell) \rvert \geq \ell ^ \prime
        \end{matrix}
    \right\}
    \],
    which roughly says that a substring of length $\ell ^ \prime$ had so many deletions that only $\ell$ or fewer bits were left in the output. For $C > 1 / q$, we have $\Pr (H) \leq \exp(-\Omega(\ell^\prime))$.

    As long as $H$ does not occur, then any spurious match in equation \eqref{eq:spurious_match} must have come from within an interval $J ^ \prime$ of length at most $\ell^\prime$.
    In other words
    \[
        \mcQ_{\ell, \lambda} \subseteq H \cup \left(\bigcup_{\substack{J^\prime \subseteq J \\ \lvert J \rvert = \ell^\prime}} \mcQ_{\ell, \lambda}(I_i, J^\prime) \right).
    \]

    Let $J^\prime_i$ be the subsegment of $J$ of length $\ell^\prime$, which maximizes $\Pr [\mcQ_{\ell, \lambda}(I_i, J^\prime_i)]$.
    By the union bound, it is easy to see that:
    \[
        \Pr [\mcQ_{\ell, \lambda}(I_i, J^\prime_i)] \geq 
        \frac{1}{\lvert J \rvert} 
        \Pr\left[ 
            \left(\bigcup_{\substack{
                J^\prime \subseteq J \\ 
                \lvert J \rvert = \ell^\prime
            }} \mcQ_{\ell, \lambda}(I_i, J^\prime) \right) 
            \right] \geq \exp (-2 c_0 C^{7/12} h(C \log n)).
    \]

    We will want to prove that w.p. $\geq 1 - n^{-2}$, at least one of the $I_i$ segments exhibits a robust bias at scale $\lambda$ and is not bad.
    We will say that a pair $I_i, J_i$ is bad, if $I_i$ does not exhibit robust bias at scale $\lambda$, or if $\Pr [\mcQ_{\ell, \lambda}(I_i, J_i)] \geq \exp (-2 c_0 C^{7/12} h(C \log n))$ 
    
    We have shown that for any bad or non-robust segment $I_i$, there exists a segment $J_i$ such that $(I_i, J_i)$ are bad.
    We will bound the probability that all the $I_i$ segments are bad by union bounding over all the options of $(J_1, \ldots, J_m)$ and the probability that the pairs $(I_i, J_i)$ are all bad.

    We first note that there are $(C \log n)^m = n^{o(1)}$ options for the assignment of the segments $J_i$, and we will prove that for any assignment of $J_i$, the probability that all the $(I_i, J_i)$ pairs are bad is $\leq n^{-3}$.

    Fix some choice of $J_1, \ldots, J_m$.
    We will show that the probability that $(I_i, J_i)$ are all bad is very small.
    If the Bernoulli random variables $B_i$ determining whether each pair was good or bad were independent of each other, the claim would follow trivially from Lemmas \ref{lem:most_robust}, \ref{lem:within_interval} and \ref{lem:cross_interval}.
    We will show that for $r = 0.01 C^{1/6} \log n / h(C\log n)$, there are indices $i_1, \ldots, i_r$ such that $B_i$ are in some sense sufficiently close to being independently distributed.
    
    We take $i_1 = 1$ and for each $k \geq 1$, let $N_k$ be
    \[
    N_k \defeq \bigcup_{j=1}^k \left(I_{i_j} \cup J_{i_j}\right).
    \]
    
    Then, choose $i_{k+1}$ such that $I_{i_{k+1}}$ is distance at least $2\ell^\prime$ from $N_k$.
    Note that the $2\ell^\prime$-neighborhood of $N_k$ intersects at most $2k \lceil 5\ell^\prime / \ell \rceil \leq 12 C^{1/6} k$ of the $I_i$, so such a choice is always possible when $k \leq r$.
    
    Let $\mathcal{G}_k$ be the $\sigma$-field generated by the bits of $\bfx$ whose positions are in $N_k$, and let $E_k$ denote the event that $(I_{i_k}, J_{i_k})$ is a bad pair.
    Note that $E_k$ is measurable with respect to $\mathcal{G}_k$.
    
    First of all, note that whether $I_{i_k}$ has clear robust bias at scale $\lambda$ is independent of $\mathcal{G}_{k-1}$, so by \thref{lem:most_robust}, we have
    \[
    \Pr \left( \substack{I_{i_k} \textup{ does not have clear}\\ \textup{robust bias at scale }\lambda } \mid \mathcal{G}_{k-1} \right) \leq e^{-\Omega(\ell/\lambda)} \leq e^{-\Omega(C^{7/12} h(C \log n) }
    \]
    
    Next, we will estimate $\Pr( \mcQ_{\ell, \lambda} (I_{i_k}, J_{i_k}) \mid \mathcal{G}_{k-1})$.
    Suppose first that $I_{i_k}$ and $J_{i_k}$ are disjoint and distance $\geq \ell^\prime$ separated.
    Then, by \thref{lem:cross_interval}, we have:
    \[
    \Pr \left( \mcQ_{\ell, \lambda} (I_{i_k}, J_{i_k}) \mid \mathcal{G}_{k-1} \right) \leq \ell^\prime e^{-10c_0 \ell /\lambda} + e^{-10 c_0 \ell^\prime} \leq e^{-9 c_0 C^{7/12} h(C \log n)}
    \]
    In the last step, we use our assumption that $ \log n \leq h(n) \leq n$, implying that for large enough $C$
    \[
    h(n) \leq \exp(\frac{1}{2} c_0 C^{7/12} h(C \log n))
    \]
    
    If instead $I_{i_k}$ and $J_{i_k}$ are within distance $\ell^\prime$ of each other, then let $J$ be the interval formed by extending $I_{i_k}$ on both sides by $2\ell^\prime$, so that $J_{i_k} \subseteq J$.
    By our construction, it is also guaranteed that $J$ is disjoint from $N_{k-1}$, and so when conditioning on $\mathcal{G}_{k-1}$, none of its bits have been determined yet.
    Then, we may apply \thref{lem:within_interval} to obtain
    \begin{equation}
        \begin{aligned}
            \Pr \left( \mcQ_{\ell, \lambda} (I_{i_k}, J_{i_k}) \mid \mathcal{G}_{k-1} \right) &\leq \Pr \left( \mcQ_{\ell, \lambda} (I_{i_k}, J) \mid \mathcal{G}_{k-1} \right)
            &\leq 3 \ell^\prime e^{-10c_0 \ell / \lambda} + e^{-10c_0\ell^\prime} \leq e^{-9 c_0 C^{7/12} h(C \log n)}
        \end{aligned}
    \end{equation}
    
    Therefore
    \[
    \Pr(E_1 \cap \ldots \cap E_r) \leq e^{-\Omega(C^{7/12} h(C \log n)) r} \leq e^{-\Omega(C^{3/4} \log n)}
    \]
    
    which is smaller than $n^{-3}$ for large enough $C$, proving our claim.
    
\end{proof}

\section{Proof of \texorpdfstring{\thref{lem:interpolation}}{Lemma \ref{lem:interpolation}}}
\label{app:interpolation}
\begin{proof}

We will prove the claim by induction on $l$. The proofs for the step and the base will be identical.

Let $F(\bfj, l, n)$ be the maximal factor between the inaccuracy of our estimate of the oracle (i.e. $\delta$) to our estimate of the $\bfj$th monomial.

We will show that the iteration step maintains
\[
F(\bfj, l, n) \leq \poly(n, 1/c)^{C(j_l+1)} F(\bfj(1:l-1), l-1, n)
\]
for some global constant $C$.
From this, we can easily derive the bound
\[
F(\bfj, l, n) \leq \poly(n, 1/c)^{C(j_l+1)}
\]

\vs

We view the polynomial $p$ as a polynomial in the last variable $z_l$ whose coefficients are themselves degree $\leq n$ polynomials in the first $l-1$ variables. In other words:
\[
p(z_1, \ldots, z_l) = p_{z_1, \ldots, z_{l-1}} (z_l) = \sum_{0 \leq i \leq n} a_i(z_1, \ldots, z_{l-1}) z_l ^i
\]

We will show that for any given point $z_1, \ldots, z_{l-1}$, we can compute $a_j(z_1, \ldots, z_{l-1})$ to within an error of at most $\poly(n, 1/c)^{(j+1)} \delta$, given $n$ queries to the oracle $D$.
This will prove our claim, because we can use our induction and these queries to $a_{j_l}(z_1, \ldots, z_{l-1})$ in order to compute the coefficient of its $\bfj(1:l-1)$th monomial.

We compute $a_j(z_1, \ldots, z_{l-1})$ using Lagrange polynomials.
For any $-n \leq i \leq n$, we define $x_i = ic / n \in [-c, c]$, and we define the Lagrange polynomials:
\[
\mathcal{L}_i = \prod_{i^\prime \neq i} \frac{x - x_{i^\prime}}{x_i - x_{i^\prime}}
\]

A commonly used fact about these polynomials is that they can be used to interpolate.
Indeed, let $g(x)$ be a single-variable polynomial of degree $\leq 2n$, and consider the polynomial $f(x) = \sum_i g(x_i) L_i (x)$.
Clearly, $f(x)$ is a polynomial of degree $2n$ and because $L_i(x_{i^\prime}) = \delta_{i,i^\prime}$, we have that $f(x_i) = g(x_i)$ at all the points in the interpolation.
Since no non-zero polynomial of degree $\leq 2n$ can have $2n+1$ roots, this implies that $f(x) = g(x)$.

Define $\lambda_{i,j}$ to be the coefficient of the $j$th monomial $x^j$ of $L_i(x)$.
If we can bound each $\abs{\lambda_{i, j}}$ from above, then using the triangle inequality, we can also bound our error for estimating $a_j$:
\begin{equation}
    \begin{aligned}
    a_j (z_1, \ldots, z_{l-1}) &= \sum_i \lambda_{i, j} p(z_1, \ldots, z_{l-1}, x_i) \\
    &= \sum_i \left(\lambda_{i, j} P(z_1, \ldots, z_{l-1}, x_i) \pm \delta \right)\\
    &= \sum_i \lambda_{i, j} P(z_1, \ldots, z_{l-1}, x_i) \pm \sum_i \abs{\lambda_{i,j}} \delta.
    \end{aligned}
\end{equation}

\vs

We begin by writing an exact formula for $\lambda_{i,j}$:
\begin{equation}
\label{eq:lambda_ij}
    \begin{aligned}
    \lambda_{i,j} = \sum_{\substack{J \subseteq \{-n, -n+1, \ldots, n\}\\ i\notin J\wedge\abs{J} = j}} \prod_{\substack{i^\prime \neq i \\ i^\prime \notin J}} \frac{-x_{i^\prime}}{x_i - x_{i^\prime}} \prod_{i^\prime \in J} \frac{1}{x_i - x_{i^\prime}}
    \end{aligned}
\end{equation}

For any $-n \leq i\neq i^\prime \leq n$, we clearly have:
\begin{equation}
\label{eq:xi_xi_prime}
    \begin{aligned}
    &\frac{\abs{x_i - x_{i^\prime}}}{\frac{c}{n}} = \abs{i - i^\prime} \in [1, 2n]
    \end{aligned}
\end{equation}

Furthermore, for all $i \neq 0$, clearly $\abs{x_i} \geq \frac{c}{n}$.

Combining \eqref{eq:lambda_ij} and \eqref{eq:xi_xi_prime}, we have:

\begin{equation}
\label{eq:bound_lambdas}
    \begin{aligned}
    \abs{\lambda_{i,j}} &\leq \sum_{\substack{J \subseteq \{-n, -n+1, \ldots, n\}\\ i\notin J\wedge\abs{J} = j}} \abs{\prod_{\substack{i^\prime \neq i \\ i^\prime \notin J}} \frac{-x_{i^\prime}}{x_i - x_{i^\prime}} \prod_{i^\prime \in J} \frac{1}{x_i - x_{i^\prime}}} \\
    &\leq \left( \begin{matrix} 2n \\ j \end{matrix} \right)  \left(\frac{c}{n}\right)^{-j} \abs{\prod_{\substack{-n \leq i^\prime \leq n \\ {i^\prime \neq 0, i}}} \frac{x_{i^\prime}}{x_i - x_{i^\prime}}} \\
    &\leq n\left(\frac{2n^2}{c}\right)^j \frac{(n!)^2}{(n+i)!(n-i)!}
    \end{aligned}
\end{equation}

All that remains is to bound the fraction of at the end of equation \eqref{eq:bound_lambdas}.
But it can be easily bounded by $1$ with the following inequality that follows from basic Combinatorics.
Let $k = \abs{i} \geq 0$. Then:
\[
\frac{(n+k)!}{n! k!} = \left( \begin{matrix} n+k \\ k \end{matrix}\right) > \left( \begin{matrix} n \\ k \end{matrix}\right) = \frac{n!}{(n-k)! k!}
\]

This implies that $\abs{\lambda_{i,j}} \leq (2n)^{2j+1} c^{-j}$, further implying that
\[
F(\bfj, l, n) \leq F(\bfj(1:l-1), l-1, n) (2n)^{2j+2} c^{-j} \leq\cdots \leq (2n)^{2j_\textup{tot}+2l} c^{-j_\textup{tot}}
\]
\end{proof}

\section{Proof of \texorpdfstring{\thref{thm:like_borwein}}{Theorem \ref{thm:like_borwein}}}
\label{app:proof_of_chase}

In this section, we will prove \thref{thm:like_borwein}.
This theorem is based on Theorem 5 of \cite{chase2021separating} and we base our proof on Chase's proof.

\vs

Let $n$ be sufficiently large, and let $\rho, \mu$ be as in \thref{thm:like_borwein}.
Let $p_{\textup{base}}$ be a polynomial in $\mcP_n ^ \mu$, and define $p(z) = p_{\textup{base}}(\rho z)$.
Showing that for some $\theta \in [-n^{-2/5},n^{-2/5}]$, we have $\lvert p(e^{i\theta}) \rvert \geq \exp(-C n^{\mu} \log^5 n)$, will yield our main claim, so we will try to lower bound the maximum of $p$ on this arc.

\vs

Let $a = n^{-2/5}$ and $r = a^{-1/2}$. Let $r_* \in [r]$ be such that $$\sum_{j=1}^{r_*} \frac{1}{\log^2(j+3)}-\sum_{j=r_*+1}^r \frac{1}{\log^2(j+3)} \in [20,21];$$ such an $r_*$ clearly exists. Let $$\begin{cases} \epsilon_j = +1 & \text{if }  1 \le j \le r_* \\ \epsilon_j = -1 & \text{if } r_*+1 \le j \le r\end{cases}.$$ Let $\lambda_a \in (1,2)$ be such that $$\sum_{j=1}^r \frac{\lambda_a}{j^2\log^2(j+3)} = 1.$$ Let $$d_j = \frac{\lambda_a}{j^2\log^2(j+3)}.$$ Define $$\wt{h}(z) = \wt{\lambda}_a\sum_{j=1}^r \epsilon_j d_j z^j,$$ where $\wt{\lambda}_a \in (1,2)$ is such that $\wt{h}(1) = 1$. Define $$h(z) = (1-a^{10})\wt{h}(z).$$ Let $$\alpha = e^{ia}, \beta = e^{-ia},$$ and $$I_t = \{z \in \C : \arg(\frac{\alpha-z}{z-\beta}) = t\}$$ for $t \ge 0$. Note that $I_0$ is the line segment connecting $\alpha$ and $\beta$ and $I_a = \{e^{i \theta} : |\theta| \le a\}$ is the set on which we wish to lower bound $p$ at some point. Let $$G_a = \{z \in \C : \arg(\frac{\alpha-z}{z-\beta}) \in (\frac{a}{2},a)\}$$ be the open region bounded by $I_{a/2}$ and $I_a$. 

\vs

We use the same choice of $h$ as \cite{chase2021separating}.
It is designed to satisfy (i) $|h(e^{2\pi i t})| \le 1-c|t|$ for $|t| > a^{1/2}$ (up to logs). In this paper, we need (ii) $|h(e^{2\pi i t})| \ge 1-Ca^2$ for $|t| \approx a$. The following lemmas are proven in Chase's paper:

\begin{lemma} [Lemma 3 of \cite{chase2021separating}]
\label{unitdisk}
For any $t \in [-\pi,\pi]$, $\wt{h}(e^{it}) \in \ol{\D}$. 
\end{lemma}

\vspace{0.5mm}

\begin{lemma} [Lemma 4 of \cite{chase2021separating}]
\label{hproperties} 
There are absolute constants $c_4,c_5,C_6 > 0$ such that the following hold for $a > 0$ small enough. First, $h(e^{2\pi i t}) \in G_a$ for $|t| \le c_4a$. Second, $|h(e^{2\pi i t})| \le 1-c_5\frac{|t|}{\log^2(a^{-1})}$ for $t \in [\frac{-1}{2},\frac{1}{2}]\setminus [-C_6a^{1/2},C_6a^{1/2}]$.
\end{lemma} 

Lemmas \ref{unitdisk} and \ref{hproperties} are used in the same manner as Chase, so we do not repeat their proofs.
However, the next lemmas are slightly adapted to our case, because we will want to evaluate our polynomial at a point $z$ with absolute value strictly lower than $1$.

\vs

Let $m = c_4^{-1}n^{2/5}, J_1 = c_5^{-1}n^{-1/5}m\log^4 n$, and $J_2 = m-J_1$.

\begin{lemma}\label{uproduct}
Suppose $u(z) = \eta z^d-\zeta$ for some $\zeta \in \partial \mathbb{D}$, integer $1 \leq d \leq a^{1/2}$ and some $\eta \in [0, 1]$. Then, for any $\delta \in [0,1)$, we have $\prod_{j=J_1}^{J_2-1} |u(h(e^{2\pi i \frac{j+\delta}{m}}))| \le \exp(Cn^{1/5}\log^5n)$. 
\end{lemma}

\begin{proof}
First note that 
\begin{equation}
\label{note} 
|u(h(e^{2\pi i \theta}))| \ge 1-\eta |h(e^{2\pi i \theta})|^d \geq 1 - (1 - a^{10})^d \ge a^{10}. 
\end{equation} 
Define $g(t) = 2\log|u(h(e^{2\pi i (t+\frac{\delta}{m})}))|$. For notational ease, we assume $\delta = 0$; the argument about to come works for all $\delta \in [0,1)$. 
Since \eqref{note} implies $g$ is $C^1$, by the mean value theorem we have \begin{align}\label{mvt}\left|\frac{1}{m}\sum_{j=J_1}^{J_2-1} g\left(\frac{j}{m}\right)-\int_{J_1/m}^{J_2/m} g(t)dt\right| &= \left|\sum_{j=J_1}^{J_2-1} \int_{j/m}^{(j+1)/m} \left(g(t)-g\left(\frac{j}{m}\right)\right)dt\right| \nonumber\\ &\le \sum_{j=J_1}^{J_2-1} \int_{j/m}^{(j+1)/m} \left(\max_{\frac{j}{m} \le y \le \frac{j+1}{m}} |g'(y)|\right)\frac{1}{m} dt \nonumber\\ &\le \frac{1}{m^2}\sum_{j=J_1}^{J_2-1}\max_{\frac{j}{m} \le y \le \frac{j+1}{m}} |g'(y)|. \end{align} Since $w \mapsto \log|u(h(w))|$ is harmonic and $\log|u(h(0))| = \log|u(0)| = 0$, we have $$\int_0^1 g(t)dt = 2\int_0^1 \log |u(h(e^{2\pi i t}))|dt = 0,$$ and therefore
\begin{equation} \label{harmonicapplication}
\left|\int_{J_1/m}^{J_2/m} g(t)dt\right| \le \left|\int_{0}^{J_1/m}g(t)dt\right|+\left|\int_{J_2/m}^1 g(t)dt\right|.
\end{equation}
Since $$a^{10} \le \left|u(h(e^{2\pi i t}))\right| \le 2$$ for each $t$, we have
\begin{equation}\label{integralbounds}
\left|\int_{0}^{J_1/m}g(t)dt\right|+\left|\int_{J_2/m}^1 g(t)dt\right| \le 20\left(\frac{J_1}{m}+(1-\frac{J_2}{m})\right)\log n \le C\frac{\log^5 n}{n^{1/5}}.
\end{equation}
By \eqref{mvt}, \eqref{harmonicapplication}, and \eqref{integralbounds}, we have $$\left|\frac{1}{m}\sum_{j=J_1}^{J_2-1} g(\frac{j}{m})\right| \le C\frac{\log^5 n}{n^{1/5}}+\frac{1}{m^2}\sum_{j=J_1}^{J_2-1} \max_{\frac{j}{m} \le t \le \frac{j+1}{m}} |g'(t)|.$$
Multiplying through by $m$, changing $C$ slightly, and exponentiating, we obtain 
\begin{equation}\label{partialproductbound}
\prod_{j=J_1}^{J_2-1} \left|u(h(e^{2\pi i \frac{j}{m}}))\right|^2 \le \exp\left(Cn^{1/5}\log^5 n + \frac{1}{m}\sum_{j=J_1}^{J_2-1} \max_{\frac{j}{m} \le t \le \frac{j+1}{m}} |g'(t)|\right).
\end{equation}
Note $$g'(t_0) = \frac{\frac{\partial}{\partial t}\Big[|u(h(e^{2\pi i t}))|^2\Big] \Big|_{t=t_0}}{|u(h(e^{2\pi i t_0}))|^2}.$$ 
We first show 
\begin{equation}\label{derivativebound} 
\frac{\partial}{\partial t}\Big[|u(h(e^{2\pi i t}))|^2\Big] \Big|_{t=t_0} \le 500 d
\end{equation} 
for each $t_0 \in [0,1]$. 
Let $\wt{d}_j = d_j$ for $j \le r_*$ and $\wt{d}_j = -d_j$ for $j > r_*$ so that $h(e^{2\pi i t}) = (1-a^{10})\sum_{j=1}^r \wt{d}_j e^{2\pi i tj}$. 
Then,
\[
\abs{u\left(h(e^{2\pi i t})\right)}^2 = \abs{\eta \left(1-a^{10}\right)^d\left(\sum_{j=1}^r \wt{d}_j e^{2\pi i jt}\right)^d-\zeta}^2
\]
\begin{equation}
\label{derivative} = \eta \left(1-a^{10}\right)^{2d}\left[\abs{\sum_{j=1}^r \wt{d}_j e^{2\pi i jt}}^2\right]^d-2\eta\Real\left[\left(1-a^{10}\right)^{d}\zeta \left(\sum_{j=1}^r \wt{d}_j e^{2\pi i jt}\right)^d\right]+1.
\end{equation}
The derivative of the first term is
\[
(1-a^{10})^{2d}\eta d\left[\abs{\sum_{j=1}^r \wt{d}_j e^{2\pi ijt} } ^ 2\right]^{d-1} \sum_{j_1,j_2=1}^r \wt{d}_{j_1}\wt{d}_{j_2}2\pi(j_1-j_2)e^{2\pi i (j_1-j_2)t}.
\]
Since
\[
\sum_{j=1}^r \abs{\wt{d}_j} = \wt{\lambda}_a \le 1 + 4 a^{1/2}
\]
and
\[
\sum_{j=1}^r \abs{j\wt{d}_j} \le 4,
\] 
we get an upper bound of $250d$ for the absolute value of the derivative of the first term of \eqref{derivative}. 
The derivative of the second term, if $\zeta = e^{i\theta}$, is
\[
2(1-a^{10})^d \eta \sum_{1 \leq j_1 , \ldots, j_d \leq r} 2\pi (j_1 + \cdots + j_d) \wt{d}_{j_1}\cdots \wt{d}_{j_d} \sin(2\pi (j_1 + \cdots + j_d) t+\theta),
\]
which is also clearly upper bounded by (crudely) $250d$. We've thus shown \eqref{derivativebound}. 

\vs

Recall $|u(h(e^{2\pi i \theta}))| \ge 1-|h(e^{2\pi i \theta})|^d\eta$. 
For $j \in [J_1,J_2] \subseteq [C_6a^{1/2}m,(1-C_6a^{1/2})m]$, we use (by Lemma \ref{hproperties})
\[
\abs{h(e^{2\pi i \frac{j}{m}})} \le 1-c_5\frac{\min(\frac{j}{m},1-\frac{j}{m})}{\log^2 n}
\]
to obtain
\[
g^\prime (t) \leq \frac{1}{m}\sum_{j=J_1}^{J_2-1}\max_{\frac{j}{m} \le t \le \frac{j+1}{m}} |g'(t)| \le \frac{1}{m}\sum_{j=J_1}^{J_2-1} \frac{500d}{\left(1 - \eta \left(1 - c_5\frac{\min(\frac{j}{m},1-\frac{j}{m})}{\log^2 n}\right)^d\right)^2}.
\]
Up to a factor of $2$, we may deal only with $j \in [J_1,\frac{m}{2}]$. 
Let $J_* = c_5 ^{-1} d^{-1} m \log^2 n$.
Note that $j \leq J_*$ implies $c_5 \frac{j}{m \log^2 n} \leq d^{-1}$ and $j \geq J_*$ implies $c_5 \frac{j}{m \log^2 n} \geq d^{-1}$.
Thus, using $(1 - x)^d \leq 1 - \frac{1}{2}xd$ for $x \leq \frac{1}{d}$, we have
\begin{equation}
\label{eq:first_part}
    \begin{aligned}
        &\frac{1}{m} \sum_{j = J_1}^{\min(J_*, m/2)} \frac{500 d}{\left(1 - \eta \left(1 - c_5\frac{j}{m\log^2 n}\right)^d\right)^2} \leq \frac{500d}{m} \sum_{j = J_1}^{\min(J_*, m/2)} \frac{1}{\left(\frac{1}{2} c_5 \left(1 - c_5\frac{j}{m\log^2 n}\right)^d\right)^2}\\
        &\;\;\;\;\;\;\;\;\;\;\;\;\;\;\;\;\;\;\;\;\;\;\;\;\;\;\;\;\;\;\;\; = \frac{2000m \log^4 n}{c_5 ^2 d} \sum_{j = J_1}^{\min(J_*, m/2)} \frac{1}{j^2}\\
        &\;\;\;\;\;\;\;\;\;\;\;\;\;\;\;\;\;\;\;\;\;\;\;\;\;\;\;\;\;\;\;\; \leq \frac{2000m \log^4 n}{c_5 ^2 d} \frac{2}{J_1} \leq C n^\mu.
    \end{aligned}
\end{equation}

Finally, since there is some $c > 0$ such that $(1 - x)^l \leq 1-c$ for all $l\in\N\setminus \{0\}$ and $x \in [l^{-1}, 1]$, using the notation $\sum_{i = a} ^b x_i = 0$ if $a > b$, we see that
\begin{equation}
\label{eq:second_part}
    \begin{aligned}
        &\frac{1}{m} \sum_{j = \min(J_*, m/2)}^{m/2} \frac{500 d}{\left(1 - \eta \left(1 - c_5\frac{j}{m\log^2 n}\right)^d\right)^2} \leq \frac{500d}{m} \sum_{j = \min(J_*, m/2)}^{m/2} c^{-2}\\
        &\;\;\;\;\;\;\;\;\;\;\;\;\;\;\;\;\;\;\;\;\;\;\;\;\;\;\;\;\;\;\;\; \leq Cd\\
        &\;\;\;\;\;\;\;\;\;\;\;\;\;\;\;\;\;\;\;\;\;\;\;\;\;\;\;\;\;\;\;\; \leq C n^\mu.
    \end{aligned}
\end{equation}

Combining \eqref{eq:first_part} and \eqref{eq:second_part}, we obtain
\[
\frac{1}{m}\sum_{j = J_1}^{J_2-1} \max_{\frac{j}{m} \leq t \leq \frac{j+1}{m}} \abs{g^\prime (t)} \leq C n^\mu.
\]
Plugging this upper bound into \eqref{partialproductbound} yields the desired result.

\end{proof}

\vs

Let $\mcQ_n$ denote all polynomials of the form $(z-\alpha)(z-\beta)p(\eta z)$ for $p \in \mcP_n$ and $\eta \in [0, 1]$. 

\vspace{1mm}

\begin{corollary}\label{productbound}
For any $q \in \mcQ_n$ and $\delta \in [0,1)$, $\prod_{j\not \in \{0,m-1\}} |q(h(e^{2\pi i \frac{j+\delta}{m}}z))| \le \exp(Cn^{1/5}\log^5 n)$. 
\end{corollary}

\begin{proof}
Take $q \in \mcQ_n$; 
say $q(z) = (z-\alpha)(z-\beta)p(\eta z)$ for $p \in \mcP_n$. 
For $j \in \{1,\dots,J_1-1\}$ and for $j \in \{J_2,\dots,m-2\}$, by Lemma \ref{unitdisk} we can bound $|q(h(e^{2\pi i \frac{j}{m}}z))| \le 4n$, to obtain
\begin{equation}\label{smallj} 
\prod_{j \not \in \{J_1,\dots,J_2-1\}} |q(h(e^{2\pi i \frac{j+\delta}{m}}))| \le (4n)^{J_1-1+m-J_2-1} \le e^{Cn^{1/5}\log^5 n}.
\end{equation} 
By applying Lemma \ref{uproduct} to $u(z) := z-\alpha$ and to $u(z) := z-\beta$ and multiplying the results, we see 
\begin{equation}\label{uterms} 
\prod_{j=J_1}^{J_2-1} |\overline{u}(h(e^{2\pi i \frac{j+\delta}{m}}))| \le e^{Cn^{1/5}\log^5 n},
\end{equation} 
where $\overline{u}(z) := (z-\alpha)(z-\beta)$. 
Let $\wt{p}(z) \in \{1,1-z^d\}$ be the truncation of $p$ to terms of degree less than $n^{1/5}$. 
Then, since Lemma \ref{hproperties} gives
\[
|\eta h(e^{2\pi i \frac{j+\delta}{m}})| \le 1-c_5\frac{\min\left(\frac{j}{m}+\delta,1-(\frac{j}{m}+\delta)\right)}{\log^2n} \le 1-c'n^{-1/5}\log^2n
\]
for $j \in \{J_1,\dots,J_2-1\}$, we see 
\begin{equation}\label{approximation} 
\left|p\hspace{-.5mm}\left(\eta h(e^{2\pi i \frac{j+\delta}{m}})\right)-\wt{p}\hspace{-.5mm}\left(\eta h(e^{2\pi i \frac{j+\delta}{m}})\right)\right| \le ne^{-c'\log^2 n} \le e^{-c\log^2n}.
\end{equation} 
Lemma \ref{hproperties} implies 
\begin{equation}\label{ptildebound} \prod_{j=J_1}^{J_2-1} |\wt{p}(\eta h(e^{2\pi i \frac{j+\delta}{m}}))| \le e^{Cn^{1/5}\log^5 n}.
\end{equation} 
By an easy argument given in \cite{chase2021separating}, \eqref{approximation} and \eqref{ptildebound} combine to give
\begin{equation}\label{pterms} 
\prod_{j=J_1}^{J_2-1} |p(\eta h(e^{2\pi i \frac{j+\delta}{m}}))| \le e^{C'n^{1/5}\log^5 n}.
\end{equation} 
Combining $\eqref{smallj}, \eqref{uterms}$, and $\eqref{pterms}$, the proof is complete. 
\end{proof}

\vspace{1.5mm}

\begin{proposition}\label{eregionlowerbound}
For any $q \in \mcQ_n$, it holds that $\max_{w \in G_a} |q(w)| \ge \exp(-Cn^{1/5}\log^5 n)$. 
\end{proposition}

\begin{proof}
Let $g(z) = \prod_{j=0}^{m-1} q(h(e^{2\pi i \frac{j}{m}}z))$. For $z = e^{2\pi i \theta}$, with, without loss of generality, $\theta \in [0,\frac{1}{m})$, we have by Lemma \ref{hproperties} and Corollary \ref{productbound} $$|g(z)| \le \left(\max_{w \in G_a} |q(w)|\right)^2\prod_{j \not \in \{0,m-1\}} |q(h(e^{2\pi i (\frac{j}{m}+\theta)}))| \le \left(\max_{w \in G_a} |q(w)|\right)^2\exp(Cn^{1/5}\log^ 5n).$$ Thus, $\left(\max_{w \in G_a} |q(w)|\right)^2\exp(Cn^{1/5}\log^ 5n) \ge \max_{z \in \partial \mathbb{D}} |g(z)| \ge |g(0)| = 1$, where the last inequality used the maximum modulus principle (clearly $g$ is analytic). 
\end{proof}

\vspace{1.5mm}

The following lemma was proven in \cite{borwein1997littlewood}. 

\vspace{1mm}

\begin{lemma}\label{hadamardregions}
Suppose $g$ is an analytic function in the open region bounded by $I_0$ and $I_a$, and suppose $g$ is continuous on the closed region between $I_0$ and $I_a$. Then, $$\max_{z \in I_{a/2}} |g(z)| \le \left(\max_{z \in I_0} |g(z)|\right)^{1/2}\left(\max_{z \in I_a} |g(z)|\right)^{1/2}.$$
\end{lemma}

\vs

\begin{proof}[Proof of \thref{thm:like_borwein}]
Take $f \in \mcP_n$, and let $g(z) = (z-\alpha)(z-\beta)f(\eta z)$. 
A straightforward geometric argument yields $$|g(z)| \le \frac{|(z-\alpha)(z-\beta)|}{1-\eta |z|} \le \frac{2}{\sin(a)} \le 3n^{2/5}$$ for $z \in I_0$. Letting $L = ||g||_{I_a}$, Lemma \ref{hadamardregions} then gives $$\max_{z \in I_{a/2}} |g(z)| \le (3Ln^{2/5})^{1/2}.$$ Since we then have $$\max_{z \in I_{a/2}\cup I_a} |g(z)| \le \max(L,(3Ln^{2/5})^{1/2}),$$ the maximum modulus principle implies $$\max_{z \in G_a} |g(z)| \le \max(L,(3Ln^{2/5})^{1/2}).$$ By Proposition \ref{eregionlowerbound}, we conclude $$\exp(-Cn^{1/5}\log^5 n) \le \max\left(L,(3Ln^{2/5})^{1/2}\right).$$ Thus, $$||f||_{\eta I_a} \ge \frac{1}{4}||g||_{I_a} = \frac{L}{4} \ge \exp(-C'n^{1/5}\log^5 n),$$ as desired.
\end{proof}

\end{document}